\documentclass[runningheads]{llncs}
\def \VersionWithComments{}

\usepackage{amsfonts}
\usepackage{xcolor}
\usepackage{graphicx}
\usepackage{amssymb}
\usepackage{amsmath}
\usepackage{stmaryrd}

\usepackage{algorithm}
\usepackage[noend]{algpseudocode}

\usepackage{hyperref}
\usepackage{cleveref}
\usepackage{xspace}
\usepackage{marginnote}
\usepackage{multirow}
\usepackage{rotating}
\usepackage{colortbl}
\usepackage{thm-restate}
\usepackage{enumitem}
\usepackage{orcidlink} 

\usepackage{pgfplots}
\usepackage{tikz}

\usetikzlibrary{arrows,automata,fit,positioning,shadows,shapes}
\usetikzlibrary{decorations,snakes,calc} 
\usetikzlibrary{patterns}
\tikzstyle{location}=[rectangle, rounded corners, minimum size=12pt, draw=black, fill=blue!10, inner sep=2pt]
\tikzstyle{pta}=[auto, ->, >=stealth']
\tikzstyle{PZG}=[auto, ->, >=stealth']
\tikzstyle{mergingFigure} = [>=stealth', node distance=1.8cm, yscale=.6]
\tikzstyle{location10}=[location, minimum size=10pt]
\tikzstyle{invariant}=[draw=black, dotted, inner sep=1pt, node distance=0] 
\tikzstyle{final}=[fill=green!70,double]
\tikzstyle{urgent}=[dotted, draw=red, very thick, fill=yellow]
\tikzstyle{bad}=[fill=red]

\tikzstyle{pzgstate} = [
	rectangle split,
	rectangle split horizontal,
	rectangle split parts=2,
	draw,
	rectangle split part align={center,left},
	rounded corners,
	minimum height=0.5 cm,
	minimum width=1.5 cm,
	thick
	]
\tikzstyle{fillred} = [ fill=red!20 ]
\tikzstyle{fillblue} = [ fill=blue!20 ]
\tikzstyle{fillyellow} = [ fill=yellow!20 ]
\tikzstyle{line} = [ draw,-latex',thick ]
\tikzstyle{highlightarrow} = [
	preaction={
		draw,
		line width=0.5mm
	}
]

\tikzstyle{urgent}=[fill=yellow, thick, dotted]
\tikzstyle{private}=[fill=red!50,thick]

\pgfdeclarepatternformonly{north west lines wide}
{\pgfqpoint{-1pt}{-1pt}}
{\pgfqpoint{10pt}{10pt}}
{\pgfqpoint{9pt}{9pt}}
{
	\pgfsetlinewidth{3pt}
	\pgfpathmoveto{\pgfqpoint{0pt}{9.1pt}}
	\pgfpathlineto{\pgfqpoint{9.1pt}{0pt}}
	\pgfusepath{stroke}
}

\tikzset{onslide/.code args={<#1>#2}{%
		\only<#1>{\pgfkeysalso{#2}}
}}

\definecolor{coloract}{rgb}{0.50, 0.70, 0.30}
\definecolor{colorclock}{rgb}{0.4, 0.4, 1}
\definecolor{colordisc}{rgb}{1, 0, 1}
\definecolor{colorloc}{rgb}{0.4, 0.4, 0.65}

\definecolor{colorparam}{rgb}{.66, 0.4, 0.0}
\definecolor{colorstate}{rgb}{1, 0.4, 0.0}

\definecolor{loccolor1}{rgb}{1, 0.6, 0.45}
\definecolor{loccolor2}{rgb}{0.45, 1, 0.45}
\definecolor{loccolor3}{rgb}{0.8, 0.8, 1}
\definecolor{loccolor4}{rgb}{1, 0.45, 1}
\definecolor{loccolor5}{rgb}{1, 1, 0.45}
\definecolor{loccolor6}{rgb}{0.45, 1, 1}
\definecolor{loccolor7}{rgb}{0.9, 0.6, 0.2}
\definecolor{loccolor8}{rgb}{0.7, 0.4, 1}
\definecolor{loccolor9}{rgb}{0.5, 1, 0.75}
\definecolor{loccolor10}{rgb}{0.8, 0.7, 0.6}
\definecolor{loccolor11}{rgb}{0.6, 0.7, 0.8}
\definecolor{loccolor12}{rgb}{0.2, 0.5, 0.9}
\definecolor{loccolor13}{rgb}{0.5, 0.9, 0.2}
\definecolor{loccolor14}{rgb}{0.9, 0.2, 0.5}
\definecolor{loccolor15}{rgb}{0.7, 0.7, 0.7}
\definecolor{loccolor16}{rgb}{0.8, 0.8, 0.5}

\newcommand{\styleact}[1]{\ensuremath{\textcolor{coloract}{\mathrm{#1}}}}
\newcommand{\styleclock}[1]{\ensuremath{\textcolor{colorclock}{\mathrm{#1}}}}

\newcommand{\styleloc}[1]{\ensuremath{\textcolor{colorloc}{\mathrm{#1}}}}
\newcommand{\styleparam}[1]{\ensuremath{\textcolor{colorparam}{\mathrm{#1}}}}

\usepackage{xcolor}
\definecolor{darkgreen}{rgb}{0.0, 0.4, 0.08}
\definecolor{lighterblack}{rgb}{.4, .4, .4}
\definecolor{colorparam}{rgb}{1, 0.6, 0.0}
\definecolor{mygreen}{rgb}{0,0.6,0}
\definecolor{mygray}{rgb}{0.5,0.5,0.5}
\definecolor{mymauve}{rgb}{0.58,0,0.82}

\definecolor{gris}{rgb}{0.6,0.6,0.6}
\definecolor{grisfonce}{rgb}{0.2,0.2,0.2}
\definecolor{turquoise}{rgb}{0, 1, 1}
\definecolor{vertfonce}{rgb}{0,0.85,0}
\definecolor{violet}{rgb}{0.8,0,0.8}
\definecolor{grispale}{rgb}{0.9, 0.9, 0.9}
\definecolor{cpale1}{rgb}{1, 0.3, 0.3}
\definecolor{cpale2}{rgb}{0.3, 1, 0.3}
\definecolor{cpale3}{rgb}{0.3, 0.3, 1}
\definecolor{cpale4}{rgb}{1, 0.3, 1}
\definecolor{cpale5}{rgb}{1, 1, 0.3}
\definecolor{cpale6}{rgb}{0.3, 1, 1}
\definecolor{cpale7}{rgb}{0.9, 0.6, 0.2}
\definecolor{cpale8}{rgb}{0.7, 0.4, 1}
\definecolor{cpale9}{rgb}{0.5, 1, 0.75}
\definecolor{cpale10}{rgb}{0.8, 0.7, 0.6}
\definecolor{cpale11}{rgb}{0.6, 0.7, 0.8}
\definecolor{cpale12}{rgb}{0.2, 0.5, 0.9}
\definecolor{cpale13}{rgb}{0.5, 0.9, 0.2}
\definecolor{cpale14}{rgb}{0.9, 0.2, 0.5}
\definecolor{cpale15}{rgb}{0.7, 0.7, 0.7}
\definecolor{cpale16}{rgb}{0.8, 0.8, 0.5}
\definecolor{bleuciel}{rgb}{0.90,0.95,1}
\definecolor{cv1}{rgb}{1, 0, 0}
\definecolor{cv2}{rgb}{0, 1, 0}
\definecolor{cv3}{rgb}{0, 0, 1}
\definecolor{cv4}{rgb}{1, 1, 0}
\definecolor{cv5}{rgb}{1, 0, 1}
\definecolor{cv6}{rgb}{0, 1, 1}
\definecolor{cv7}{rgb}{0.8, 0.6, 0.4}
\definecolor{cv8}{rgb}{0.5, 0.5, 1}
\definecolor{cv9}{rgb}{0.55, 0.75, 0.35}
\definecolor{cv10}{rgb}{1, 0.6, 0.1}
\definecolor{cv11}{rgb}{0.6, 0.7, 0.8}
\definecolor{cv12}{rgb}{0.2, 0.5, 0.9}
\definecolor{cv13}{rgb}{0.5, 0.9, 0.2}
\definecolor{cv14}{rgb}{1, 0.3, 0.5}
\definecolor{cv15}{rgb}{0.7, 0.7, 0.7}
\definecolor{cv16}{rgb}{0.8, 0.8, 0.5}
\definecolor{cvorange}{rgb}{1,.8,0.5}

\newcommand{\marginX}{\marginnote{\huge{\quad\textbf{!}\quad}}}

\ifdefined\VersionWithComments
  \newcommand{\mbdj}[1]{\textcolor{purple}{\marginX{}[\textbf{Mikael}: #1]}}
  \newcommand{\bfi}[1]{\textcolor{orange}{\marginX{}[\textbf{Baptiste}: #1
  ]}}
  \newcommand{\lp}[1]{\textcolor{green!70!black}{\marginX{}[\textbf{Laure}: #1 ]}}
  \newcommand{\jvdp}[1]{\textcolor{blue!40}{\marginX{}[\textbf{Jaco}:
  #1 ]}}

\else
  \newcommand{\mbdj}[1]{}
  \newcommand{\bfi}[1]{}
  \newcommand{\lp}[1]{}
  \newcommand{\jvdp}[1]{}
\fi

\newcommand{\eg}{e.g.\xspace}
\newcommand{\ie}{i.e.\xspace}

\newcommand{\game}{\ensuremath{G} }

\newcommand{\LocSet}{\ensuremath{L}}
\newcommand{\loc}{\ensuremath{\ell} }

\newcommand{\ClockSet}{\ensuremath{X}}

\newcommand{\ParamSet}{\ensuremath{P}}

\newcommand{\TransSet}{\ensuremath{T} }
\newcommand{\trans}{\ensuremath{t} }

\newcommand{\GuardSet}{\ensuremath{\mathcal{G}} }
\newcommand{\guard}{\ensuremath{g} }

\newcommand{\SubsetSet}[1]{\ensuremath{\mathcal{P}(#1)} }

\newcommand{\LabelSet}{\ensuremath{Act} }

\newcommand{\Inv}{\ensuremath{Inv}}

\newcommand{\val}{\ensuremath{v} }
\newcommand{\ClockValSet}[1]{ \ensuremath{\mathbb{R}_{\geq 0}^{#1}}}
\newcommand{\ParamValSet}[1]{ \ensuremath{\mathbb{Q}_{\geq 0}^{#1}}}
\newcommand{\ValSet}{\ensuremath{V}}

\newcommand{\state}{\ensuremath{s} }
\newcommand{\StateSpace}{\ensuremath{\mathbb{S}}}
\newcommand{\StateSet}{\ensuremath{S}}

\newcommand{\ParamLinearTerm}{\ensuremath{plt} }

\newcommand{\ZoneFormula}{\phi}

\newcommand{\delay}{\ensuremath{\delta} }

\newcommand{\RunSet}{\ensuremath{\mathcal{R}} }
\newcommand{\run}{\ensuremath{r} }

\newcommand{\hist}{\ensuremath{h} }

\newcommand{\TargetSet}{\ensuremath{{R}} }

\newcommand{\zone}{\ensuremath{Z} }

\newcommand{\strat}{\sigma}

\newcommand{\SymbState}{\ensuremath{\xi} }
\newcommand{\SymbStateSet}{\ensuremath{\Xi} }

\newcommand{\TempPred}[1]{ \ensuremath{ #1^{\swarrow} } }
\newcommand{\TempSucc}[1]{ \ensuremath{ #1^{\nearrow} } }

\newcommand{\ProjectParam}[1] { \ensuremath{#1{\downarrow_P}}}

\newcommand{\imitator}{\textsc{Imitator}\xspace}

\newcommand{\Winning}{\ensuremath{W}}
\newcommand{\Win}{\ensuremath{W_{temp}}}

\newcommand{\Cover}{\ensuremath{\mathit{Cover}}}
\newcommand{\Depends}{\ensuremath{\mathit{Depends}}}
\newcommand{\SafePred}{\ensuremath{\mathit{SafePred}}}
\newcommand{\NewWin}{\ensuremath{\mathit{NewWin}}}
\newcommand{\WinningMoves}{\ensuremath{\mathit{WinningMoves}}}
\newcommand{\LosingMoves}{\ensuremath{\mathit{LosingMoves}}}

\newcommand{\WinningParam}{\ensuremath{\mathit{WinningParam}}}
\newcommand{\Uncontrollable}{\ensuremath{\mathit{Uncontrollable}}}
\newcommand{\Controllable}{\ensuremath{\mathit{Controllable}}}
\newcommand{\NewLosing}{\ensuremath{\mathit{NewLosing}}}
\newcommand{\Explored}{\ensuremath{\mathit{Explored}}}
\newcommand{\WaitingExplore}{\ensuremath{\mathit{WaitingExplore}}}
\newcommand{\WaitingUpdate}{\ensuremath{\mathit{WaitingUpdate}}}
\definecolor{ColorLosingProp}{HTML}{1B9E77}
\definecolor{ColorCumuPrune}{HTML}{D95F02}
\definecolor{ColorCovPrune}{HTML}{7570B3}
\definecolor{ColorInc}{HTML}{E7298A}

\newcommand{\ColorBestResult}{green!20}

\newcommand{\ControllableDeadlock}{\ensuremath{\mathit{controller\_deadlock}}}

\crefname{definition}{Def.}{Defs.}
\crefname{theorem}{Thm.}{Thms.}
\crefname{lemma}{Lem.}{Lemmas}
\crefname{line}{Line}{Lines}
\crefname{algorithm}{Alg.}{Algs.}
\crefname{figure}{Fig.}{Figs.}
\crefname{appendix}{Appendix}{Appendices}
\crefname{section}{Sec.}{Sections}

\usetikzlibrary{fit,tikzmark}

\newcounter{tmkcount}
\tikzset{%
    tikzmark suffix={-\thetmkcount},%
    defaultCodeBox/.style={draw=red}%
}

\newcommand{\drawCodeBox}[4]{%
    \begin{tikzpicture}[remember picture,overlay]
        \coordinate (start) at ([yshift=1.4ex]pic cs:#2);
        \coordinate (middle) at (pic cs:#3);
        \coordinate (end) at ([yshift=-0.2ex]pic cs:#4);
        \node[inner sep=2pt,#1,fit=(start) (middle) (end)] {};
    \end{tikzpicture}%
}

\newcommand{\BeginBox}[1]{%
    \drawCodeBox{#1}{beginCB}{middleCB}{endCB}%
    \tikzmark{beginCB}\tikzmark{middleCB}%
}

\newcommand{\SetBoxEast}{%
    \unskip%
    \tikzmark{middleCB}%
}

\newcommand{\EndBox}{%
    \unskip%
    \tikzmark{endCB}%
    \stepcounter{tmkcount}%
}

\newcommand{\BoxedState}[1][defaultCodeBox]{\State\BeginBox{#1}\ignorespaces}

\algdef{S}[WHILE]{BoxedWhile}[2][defaultCodeBox]{\BeginBox{#1}\algorithmicwhile\ #2\ \algorithmicdo}%
\algdef{S}[FOR]{BoxedFor}[2][defaultCodeBox]{\BeginBox{#1}\algorithmicfor\ #2\ \algorithmicdo}%
\algdef{S}[FOR]{BoxedForAll}[2][defaultCodeBox]{\BeginBox{#1}\algorithmicforall\ #2\ \algorithmicdo}%
\algdef{S}[LOOP]{BoxedLoop}[1][defaultCodeBox]{\BeginBox{#1}\algorithmicloop}%
\algdef{S}[REPEAT]{BoxedRepeat}[1][defaultCodeBox]{\BeginBox{#1}\algorithmicrepeat}%
\algdef{S}[IF]{BoxedIf}[2][defaultCodeBox]{\BeginBox{#1}\algorithmicif\ #2\ \algorithmicthen}%
\algdef{C}[IF]{IF}{BoxedElsIf}[2][defaultCodeBox]{\BeginBox{#1}\algorithmicelse\ \algorithmicif\ #2\ \algorithmicthen}%
\algdef{Ce}[ELSE]{IF}{BoxedElse}{EndIf}[1][defaultCodeBox]{\BeginBox{#1}\algorithmicelse}%

\def\orcidID#1{\smash{\href{https://orcid.org/#1}{\protect\raisebox{-1.25pt}{\protect\includegraphics{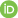}}}}}

\title{On-The-Fly Algorithm for Reachability in\\
    Parametric Timed Games (Extended Version)
    \thanks{This work was partially supported by
        CNRS international PhD programme,
        the CNRS International Research Network CLoVe
        and Innovationsfonden Danmark's DIREC project SIoT (Secure Internet of Things).
        }
    }

\author{
Mikael Bisgaard Dahlsen-Jensen \inst{1}\orcidlink{0000-0003-0641-7635}
\and
Baptiste Fievet \inst{2}\orcidlink{0000-0002-4925-1105}
\and \\
Laure Petrucci \inst{2}\orcidlink{0000-0003-3154-5268}
\and
Jaco van de Pol \inst{1}\orcidlink{0000-0003-4305-0625}
}

\institute{
Aarhus University, Aarhus, Denmark
\\\email{\{mikael,jaco\}@cs.au.dk}
\and
LIPN, CNRS UMR 7030, Université Sorbonne Paris Nord,
Villetaneuse, France
\\ \email{\{Baptiste.Fievet, Laure.Petrucci\}@lipn.univ-paris13.fr}
}

\authorrunning{Dahlsen-Jensen, Fievet, Petrucci, van de Pol}

\titlerunning{OTF Algorithm for Reachability in PTG}

\newcommand{\temptrans}{\to^{\delta}}
\newcommand{\disctrans}{\to^{t}}

\begin{document}

\maketitle

\begin{abstract}
    Parametric Timed Games (PTG) are an extension of the model of Timed Automata. They allow
    for the verification and synthesis of real-time systems, reactive to their environment
    and depending on adjustable parameters. Given a PTG and a reachability objective,
    we synthesize the values of the parameters such that the game is winning for the controller.
    We adapt and implement the On-The-Fly algorithm for parameter synthesis for PTG. 
    Several pruning heuristics are introduced, to improve termination and speed of the algorithm.
    We evaluate the feasibility of parameter synthesis for PTG on two large case studies.
    Finally, we investigate the correctness guarantee of the algorithm: though the problem is undecidable, 
    our semi-algorithm produces all correct parameter valuations ``in the limit''.
\end{abstract}

\section{Introduction}
\label{sec:intro}

The seminal model of Timed Automata (TA)~\cite{Intro-TA} equips finite automata with real-valued clocks,
to verify real-time reactive systems. Numerous extensions of TA have been proposed.
Timed Games (TG)~\cite{Classic-TG} distinguish controllable and uncontrollable actions, to study the interaction
of a controller with its environment (\eg{} the plant, an attacker, or a system-under-test).
Here, we focus on reachability objectives, which require a strategy for the controller to schedule controllable actions
such that --- no matter which and when uncontrollable actions are executed by the environment ---
a desirable state is reached.

Since precise timing constraints are not always known, one might replace concrete values by symbolic parameters,
to study a whole family of timed systems. This leads to the model of Parametric Timed Automata 
(PTA)~\cite{DBLP:conf/stoc/AlurHV93}. The problem is to find (some or all) values 
for the parameters such that the system satisfies a desired property. Most problems on PTA
are undecidable~\cite{DBLP:journals/sttt/Andre19}, in particular the reachability problem.
Several decidable fragments are known, \eg{} by restricting the number of clocks or
the positions of the parameters, as in L/U PTA~\cite{DBLP:journals/jlp/HuneRSV02}.

This paper tackles the parameter synthesis problem for Parametric Timed Games 
(PTG)~\cite{DBLP:conf/wodes/JovanovicFLR12} with reachability objectives.
We provide the first implementation of a semi-algorithm for PTG parameter synthesis.
It operates on-the-fly, \ie{} it starts solving the game while the symbolic state
space is being generated.
To avoid the generation of the full, potentially infinite, state space,
we also implement several state space reductions. These improve the termination and efficiency of parameter synthesis. 
In particular, we lift inclusion/subsumption from TA to PTG, generalize coverage pruning
and losing state propagation from TG to PTG, and we port cumulative pruning from PTA to PTG.

Interestingly, unlike the situation in PTA~\cite{DBLP:conf/tacas/AndreAPP21} and TG~\cite{OTF-TG},
the algorithm for PTG is not guaranteed to terminate, even if the symbolic state space is finite.
But we claim that if the algorithm terminates, it produces
the precise constraints under which there exists a winning strategy. If the algorithm does not terminate,
the stronger guarantee holds, that (in the limit) it produces all valid parameter valuations, 
provided the waiting list is handled fairly. 

The implementation allows us to study the feasibility of parameter synthesis for
larger case studies in PTG. In particular, we synthesize parameters for the correctness
of a game version of the Bounded Retransmission Protocol~\cite{BRP} and a parametric version
of the Production Cell~\cite{ProdCellCaseStudy2,OTF-TG}. We measure the effectiveness of the individual pruning 
heuristics on these case studies.
It appears that the state space reduction techniques are essential for feasible parameter synthesis.

\paragraph{Related Work.}
For TG, Maler et al.~\cite{Classic-TG} proposed a strategy synthesis algorithm
based on classical reachability games, handling the uncountable set of clock
values using symbolic regions. Cassez et al.~\cite{OTF-TG} improved the efficiency
of TG strategy synthesis by an on-the-fly algorithm, and working with symbolic zones,
represented by DBMs as implemented in UPPAAL Tiga~\cite{DBLP:conf/cav/BehrmannCDFLL07}.
Previous work on PTG initially focused on 
decidable subcases, like the case for bounded integers~\cite{DBLP:conf/atva/JovanovicLR13}
and the fragment of L/U PTG~\cite{DBLP:conf/wodes/JovanovicFLR12,Classic-PTG}.
The latter two papers also provide semi-algorithms for general PTG, either based
on backward fixed points~\cite{Classic-PTG}, or an on-the-fly algorithm~\cite{DBLP:conf/wodes/JovanovicFLR12},
directly extending the work on Timed Games~\cite{OTF-TG}. That paper leaves an implementation
of the algorithm (and hence an evaluation on larger case studies) as future work.
Our implementation extends the infrastructure of \imitator~\cite{DBLP:conf/cav/Andre21},
which so far could only handle PTA. The symbolic data structure is based on Parma's convex 
Polyhedra Library~\cite{DBLP:journals/scp/BagnaraHZ08}.

\paragraph{Contributions.}
(1) We provide the first implementation of a parameter synthesis algorithm for PTG (\cref{sec:algo}), 
and integrate this on-the-fly algorithm in the \imitator toolset~\cite{DBLP:conf/cav/Andre21} (\cref{sec:expe}). 

(2) We devise and implement several pruning heuristics to speed up
parameter synthesis (\cref{sec:optimizations}).

(3) We evaluate the feasibility of parameter synthesis for PTG on two large case
studies, and measure the effect of the various pruning techniques (\cref{sec:expe}).

(4) We carefully introduce the model (\cref{sec:model}) and solution principles (\cref{sec:solve}),
pointing out several semantic subtleties, and find that the semi-algorithm yields all valid 
parameters in the limit (\cref{sec:algo}).

\section{Model of Parametric Timed Games}
\label{sec:model}

A Parametric Timed Game (PTG) is a structure based on timed automata (TA). Similarly
to classical automata, it is composed of locations connected by discrete transitions.
Moreover, it is equipped with clocks.
Locations are associated a condition on clock valuations (invariant) that must
be satisfied while staying in the location.
An action in a timed automaton is either to take a discrete transition or to let some
time pass.
Discrete transitions have a guard that must be
satisfied in order to take the transition. In a parametric setting, these conditions
use linear terms over clocks and parameters. Parameters 
hold an unspecified value, and remain constant during a run. A discrete transition
also has a subset of clocks which are reset when the transition is taken.

In a two-player timed game, discrete transitions are partitioned
between controllable transitions and uncontrollable environment transitions.

\begin{definition}[PTG]    
A \emph{Parametric Timed Game} is a tuple of the form
$\game = (\LocSet , \ClockSet, \ParamSet, \LabelSet, \TransSet_c, \TransSet_u, \loc_0, \Inv)$ such that
\begin{itemize}[noitemsep, topsep=0pt]
    \item $\LocSet$, $\ClockSet$, $\ParamSet$, $\LabelSet$ are sets of \emph{locations}, \emph{clocks},
        \emph{parameters}, \emph{transition labels}.
    \item $\TransSet = \TransSet_c \cup \TransSet_u$ is the set of \emph{transitions} in
        $\LocSet \times \GuardSet(\ClockSet,\ParamSet) \times \LabelSet \times \SubsetSet
        {\ClockSet} \times \LocSet$, partitioned into sets $\TransSet_c$ of \emph{controllable}
        and $\TransSet_u$ of \emph{uncontrollable} transitions
        of the form $(\styleloc{\loc},\guard,\styleact{a},
        Y,\styleloc{\loc'})$; \styleloc{\loc}, \styleloc{\loc'} are source
        and target locations; $\guard\in \GuardSet(\ClockSet,\ParamSet)$ is the guard
        (see \cref{def:guard}); \styleact{a}
        is the label; $Y$ is the set of clocks to reset.
    \item \styleloc{\loc_0} is the initial location.
    \item $\Inv \; : \; \LocSet \to \GuardSet(\ClockSet,\ParamSet)$ associates an
        \emph{invariant} with each location.
\end{itemize}
\end{definition}

\begin{example}
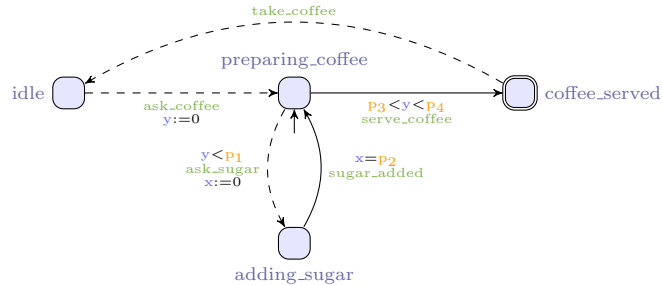
\begin{figure}[b] \centering\scriptsize
    \begin{tikzpicture}[>=stealth', node distance=3cm,initial text=]
        \node[location,initial below,label={above:\styleloc{preparing\_coffee}}] (pc) {};
        \node[location,accepting,label={right:\styleloc{coffee\_served}},right of=pc]
            (cs) {};
        \node[location,label={below:\styleloc{adding\_sugar}},below of=pc,
            node distance=2cm] (as) {};
        \node[location,label={left:\styleloc{idle}},left of=pc] (idle) {};
        \draw[->,dashed] (idle) -- node [below] {$\substack{\styleact{ask\_coffee}
        \\ \styleclock{y}:=0}$} (pc);
        \draw[->,dashed] (pc) edge [bend right] node [left]
            {$\substack{\styleclock{y}<\styleparam{p_1} \\ \styleact{ask\_sugar} \\
                \styleclock{x}:=0}$} (as);
        \draw[->,dashed] (cs) edge [bend right] node [above]
            {$\substack{\styleact{take\_coffee}}$} (idle);
        \draw[->] (pc) -- node [below] {
                            $\substack{\styleparam{p_3}<\styleclock{y}<\styleparam{p_4} \\
                                \styleact{serve\_coffee}}$} (cs);
        \draw[->] (as) edge [bend right] node [right]
            {$\substack{\styleclock{x}=\styleparam{p_2} \\ \styleact{sugar\_added}}$}
            (pc);
    \end{tikzpicture}
    \caption{Parametric Timed Game of the coffee machine.\label{fig:coffee}}
\end{figure}
\Cref{fig:coffee} shows the example of a coffee machine. The controller represents
the coffee machine and the environment represents the user. Uncontrollable transitions
are depicted as dashed arcs.
From \styleloc{idle}, the user can \styleact{ask\_coffee}.
It resets clock \styleclock{y} that will measure
the time since the demand. The machine is then \styleloc{preparing\_coffee}. Action
\styleloc{serve\_coffee} can happen after \styleparam{p_3} (parameter featuring the
time to pour the coffee) and no later than \styleparam{p_4} after the request. While
the coffee is being prepared, the user may \styleact{add\_sugar}. Adding sugar does
not interrupt the pouring of the coffee and lasts \styleparam{p_2}. The coffee cannot
be served while sugar is being added. A situation that may arise is that sugar is
being added to the coffee when the time limit \styleparam{p_4} is met, making it
impossible for the coffee to be served on time. To avoid this issue, \styleact{ask\_sugar}
is disabled after waiting \styleparam{p_1}.

Our goal is to synthesize the constraints on parameters \styleparam{p_1} to \styleparam{p_4}
for the coffee to be timely served. Hence, the initial location is
set to \styleloc{preparing\_coffee}, with both clocks at $0$.
One possible solution to the problem is $\styleparam{p_1}+\styleparam{p_2}\leq\styleparam{p_4} \wedge \styleparam{p_3}<\styleparam{p_4}$.
\end{example}

\subsection{Semantics of Parametric Timed Games}

A \emph{state} of a PTG consists of a location and a valuation of clocks and parameters.

\begin{definition}[valuations]
    A \emph{clock valuation} is a function $\val_\ClockSet \in \ClockValSet
    {\ClockSet}$ assigning a positive real value to each clock.
    A \emph{parameter valuation} $\val_\ParamSet \in \ParamValSet
    {\ParamSet}$ assigns a positive rational value to each parameter.
    A \emph{valuation of the game} $\game$ is a pair $\val = (\val_\ClockSet,\val_\ParamSet)$. 
    The set of all valuations of the game is denoted
    $\ValSet = \ClockValSet{\ClockSet} \times \ParamValSet{\ParamSet}$.
\end{definition}

A \emph{guard} is a constraint that can be satisfied by some valuations of the game.

\begin{definition}[linear terms]
A \emph{linear term} over $\ParamSet$ is a term defined by the following grammar:
$\ParamLinearTerm \; := \; k \; | \; k\styleparam{p} \; | \; \ParamLinearTerm + \ParamLinearTerm$
where $k \in \mathbb{Q}$ and $\styleparam{p} \in \ParamSet$.
\end{definition}

\begin{definition}[guards]
\label{def:guard}
The set of \emph{guards} $\GuardSet(\ClockSet,\ParamSet)$ is the set of formulas defined inductively by the following grammar:
\[\ZoneFormula \; := \top \; | \; \ZoneFormula \land \ZoneFormula
\; | \; \styleclock{x} \sim \ParamLinearTerm
\; | \; \ParamLinearTerm' \sim \ParamLinearTerm\enspace,\]
where $\styleclock{x} \in \ClockSet$,
${\sim}\in\{ <; \leq; =; \geq; > \}$ and $\ParamLinearTerm$, $\ParamLinearTerm'$ are
linear terms over $\ParamSet$.
\end{definition}

We now introduce the notion of \emph{zone} which will be used to solve a PTG.

\begin{definition}[zones]\label{def:zone}
The set of \emph{parametric zones} $\mathcal{Z}(\ClockSet,\ParamSet)$ is the set of formulas defined inductively by the following grammar:
\[\ZoneFormula \; := 
\; \top
\; | \; \ZoneFormula \land \ZoneFormula \; | \; \styleclock{x} \sim \ParamLinearTerm 
\; | \; \styleclock{x}-\styleclock{y} \sim \ParamLinearTerm
\; | \; \ParamLinearTerm' \sim \ParamLinearTerm\enspace,\]
where $\styleclock{x}, \styleclock{y}\in\ClockSet$, ${\sim} \in \{ <; \leq; =; \geq;
> \}$ and $\ParamLinearTerm$ and $\ParamLinearTerm'$ are linear terms over $\ParamSet$.
\end{definition}

Function $\val_\ParamSet$ is naturally extended to linear terms on parameters, by
replacing each parameter in the term with its valuation.
With $\val \models \ZoneFormula$, we denote that
valuation $\val = (\val_\ClockSet, \val_\ParamSet)$ \emph{satisfies} a guard or
a zone $\ZoneFormula$, which is defined in the expected manner.
Zones, guards and invariants can also be seen as a convex set in the space of valuations of the game by
considering those valuations that satisfy the condition. 

Transitions modify clock valuations by letting time pass or resetting clocks.

\begin{definition}[time delays]
    Let $\val = (\val_\ClockSet,\val_\ParamSet)$ be a valuation of the game
    and $\delay \geq 0$ a delay.
\begin{itemize}[noitemsep, topsep=0pt]
\item
    $\forall \styleclock{x}\in\ClockSet: (\val_\ClockSet+\delay)(\styleclock{x}) =
    \val_\ClockSet(\styleclock{x}) + \delay$
\item
    $\val+\delay = (\val_\ClockSet +\delay , \val_\ParamSet)$
\end{itemize}
\end{definition}
\begin{definition}[clock resets]
    Let $\val = (\val_\ClockSet,\val_\ParamSet)$ be a valuation of the game
    and $Y \subseteq \ClockSet$.
    $\val_\ClockSet[Y:=0]$ is the valuation obtained by resetting the clocks in $Y$,
    \ie{}:
\begin{itemize}[noitemsep, topsep=0pt]
\item
    $\forall \styleclock{x}\in Y: \val_\ClockSet[Y:=0](\styleclock{x})=0$ and
    $\forall \styleclock{x}\in\ClockSet\setminus Y:
        \val_\ClockSet[Y:=0](\styleclock{x})=\val_\ClockSet(\styleclock{x})$
\item    
    $\val[Y:=0]= (\val_\ClockSet[Y:=0], \val_\ParamSet)$
\end{itemize}
\end{definition}

We can now define the semantics of a Parametric Timed Game.
\begin{definition}[state]
A \emph{state} of a PTG is a pair $(\styleloc{\loc},\val)$ where $\styleloc{\loc}$
is a location and $\val$ a valuation of the game satisfying its invariant:
$\val \models\Inv(\styleloc{\loc})$. The \emph{state space} is then
$\StateSpace = \{ (\styleloc{\loc},\val)\in \LocSet \times \ValSet \; | \;
\val \models \Inv(\styleloc{\loc}) \}  = \underset{\styleloc{\loc} \in \LocSet}{\bigcup}
\; \{\styleloc{\loc}\} \times \Inv(\styleloc{\loc})\enspace.$
\end{definition}

From a state in this state space, timed and discrete transitions can happen.

\begin{definition}[timed and discrete transitions]
Let $ \delay \in \mathbb{R}_{\geq 0}$ be a time delay. A \emph{timed transition} is
a relation ${\temptrans} \in \StateSpace \times \StateSpace$ s.t.
$\forall (\styleloc{\loc},\val),(\styleloc{\loc'},\val')\in \StateSpace:
(\styleloc{\loc},\val)\temptrans (\styleloc{\loc'},\val')$ iff
$\styleloc{\loc} = \styleloc{\loc'}$ and $\val' = \val + \delay$.

Let $\trans = (\styleloc{\loc},\guard,\styleact{a},Y,\styleloc{\loc'}) \in \TransSet$
be a transition. A \emph{discrete transition} is a relation
${\disctrans} \in \StateSpace \times \StateSpace$ s.t.
$\forall (\styleloc{\loc},\val), (\styleloc{\loc'},\val')\in \StateSpace:
(\styleloc{\loc},\val) \disctrans (\styleloc{\loc'},\val')$ iff
$\val \models \guard$ and $ \val' =  \val[Y:=0]$.
\end{definition}

Let $\vec{0}$ be the clock valuation where all clocks have value $0$.
The set of possible initial states of the PTG is
$\SymbState_0 = \{ (\styleloc{\loc_0}, (\vec{0}, \val_\ParamSet)) \; | \;
\val_\ParamSet \in \ParamValSet{\ParamSet}:
(\vec{0}, \val_\ParamSet) \models \Inv(\styleloc{\loc_0}) \}$.

\begin{definition}[run]\label{def:run}
A \emph{run} of the PTG $G$ is a finite or infinite sequence of states
$\state_0 \state_1 \state_2 \ldots$ s.t. $\state_0 \in \SymbState_0$ and
$\forall i \in \mathbb{N}, \; \state_{2i} \temptrans \state_{2i+1} \disctrans \state_{2i+2}$.
$\RunSet(\game)$denotes the set of runs, and $\RunSet(\game)(\state)$ the set of those
starting from state $\state$.
\end{definition}

A run alternates between (potentially null) delays and discrete transitions,
avoiding runs that let only time pass. However, there might still be Zeno runs
where infinitely many discrete transitions are taken in a finite amount of time.
When there is no ambiguity, we omit $\game$ in the notations.

\begin{example}
Let us consider again the coffee machine in \cref{fig:coffee}. Assume the parameter
valuations are: $\val_\ParamSet(\styleparam{p_1}) = 5$,
$\val_\ParamSet(\styleparam{p_2}) = 2$, $\val_\ParamSet(\styleparam{p_3}) = 5$ and
$\val_\ParamSet(\styleparam{p_4}) = 6$.
Let $\val_\ClockSet=(\val_\ClockSet(\styleclock{x}),\val_\ClockSet(\styleclock{y}))$.
We get the sequence:
$(\styleloc{preparing\_coffee}, ((0,0), \val_\ParamSet))\to^4 
(\styleloc{preparing\_coffee}, ((4,4), \val_\ParamSet)) \to^{\styleact{ask\_sugar}}
(\styleloc{adding\_sugar}, ((0,4), \val_\ParamSet)) \to^2\\
(\styleloc{adding\_sugar}, ((2,6), \val_\ParamSet)) \to^{\styleact{sugar\_added}}
(\styleloc{preparing\_coffee}, ((2,6), \val_\ParamSet))\enspace .$
\end{example}

\begin{definition}[history]
A \emph{history} is a finite prefix of a run. The set of histories of game $\game$
is denoted $\mathcal{H}(\game)$, and those starting in state $\state$ by $\mathcal{H}
(\game)(\state)$. 
\end{definition}

The notion of coverage allows for capturing all states that can occur up to some time,
without a discrete transition.

\begin{definition}[coverage]
Let $\state , \state' \in \StateSpace$ and $\delay \geq 0$ such that $ \state \temptrans
\state'$. The \emph{coverage} of the timed transition is the set of intermediate states traversed: \\
$\Cover( \state \to^{\delay} \state') = \{\state'' \in \StateSpace \; | \; \exists \delay' : 0\leq \delay' \leq \delay \; \land \; \state \to^{\delay'} \state''\}\enspace.$

The coverage of state $\state$ is the set of states obtained from $\state$ with timed
transitions only:
$\Cover(\state) = \{\state' \in \StateSpace \; | \; \exists \delta \geq 0 \; \state \temptrans \state' \}\enspace.$

The coverage of a run $\run = \state_0 \state_1 \state_2 \ldots$ is the union of the coverage of its timed transitions. When finite, it
includes the coverage of its last state $ls(\run)$ :
$\Cover(\run) = \Big( \underset{i\in \mathbb{N}}{\bigcup}\;  \Cover(\state_{2i} \temptrans \state_{2i+1}) \Big) \cup \Cover\big(ls(\run)\big)\enspace.$

\end{definition}

\begin{definition}[reachability objective and winning runs]
Let $\TargetSet\subseteq\LocSet$ be a \emph{reachability objective}.
The set of \emph{winning runs} $\Omega_{Reach}(\TargetSet)$ is the subset of runs that visit $\TargetSet$:
$\Omega_{Reach}(\TargetSet) = \{ \run \in \RunSet \; | \; \exists\,\styleloc{\loc}\in \TargetSet, \exists
\val\in\ValSet: (\styleloc{\loc},\val)\in \Cover(\run)\}\enspace.$
\end{definition}

\begin{example}
In the coffee machine, the objective is to reach from the initial location
\styleloc{prepare\_coffee} the location \styleloc{coffee\_served}. The reachability
objective is thus $\TargetSet = \{\styleloc{coffee\_served}\}$, and
the set of winning runs is $\Omega_{Reach}(\{\styleloc{coffee\_served}\})\enspace.$
\end{example}

\subsection{Strategies in Parametric Timed Games}

We introduce a definition of a strategy that deviates from \cite{OTF-TG}, where at 
each moment, a player decides to either wait, or take a discrete transition. So
their strategy returns values in $T \cup \{wait\}$. The problem with their
strategy is that it is not always clear what should happen: for instance,
given a delay $\delay$, a history $h=\state_0\temptrans \state_1$ and a strategy $\strat$, where $\strat(h)=wait$ for $0<\delay\leq 1$ and $\strat(h)=t$ for $\delay > 1$, it is not clear
when transition $t$ happens: there is no minimal $\delay>1$. Although this works formally, it is less clear what the allowed behaviour of the winning player is precisely. For that reason,
in our definition of strategy, players must decide in advance which delay they will
take. This makes the definition more constructive, clarifying what move the winning
player will actually take (\ie{} perform an action or decide to wait for some particular time) and in the end simplifies the definition of what is winning.

Furthermore, following \cite{OTF-TG}, the definition of strategy is asymmetric for controller and environment:
If both wish to do a discrete transition, we provide priority to the environment;
this corresponds to the safest situation from a software controller point of view.
Another subtle asymmetry is that the controller cannot assume that the environment will take some
uncontrollable transition, even when waiting any longer would violate the location invariant.
While this is in line with the formal definition of strategy in TG \cite{OTF-TG}, experiments with 
UPPAAL Tiga~\cite{DBLP:conf/cav/BehrmannCDFLL07} reveal that in that tool, an uncontrollable discrete transition is actually forced when
reaching the boundary of violating an invariant. 

\begin{definition}[strategy]
A \emph{controller strategy} $\strat_c$ (resp. \emph{environment strategy} $\strat_e$)
models decision-making. It is a function, depending on a history, deciding either to wait
some amount of time (possibly infinite) or to take a discrete transition:
$\strat_c : \mathcal{H} \to \mathbb{R}_{\geq 0}^{\infty} \cup \TransSet_c$, $\strat_e:
\mathcal{H} \to \mathbb{R}_{\geq 0}^{\infty} \cup \TransSet_u$ s.t.\ $\forall h \in
\mathcal{H}$ and $\sigma\in\{\strat_c,\strat_e\}$,
\begin{enumerate}[noitemsep, topsep=0.4ex]
\item \label{tavail} If $\sigma(h) =
(\styleloc{\loc}, g, \styleact{a},Y, \styleloc{\loc'}) \in \TransSet$\\
 then $ls(h) = (\styleloc{\loc},\val)$  such that 
 $\val \models g$ and $\val[Y:=0] \models \Inv(\styleloc{\loc'})$

\item \label{altern}If $( \sigma(h) = \delay \in \mathbb{R}_{\geq 0})$ and the transition $\to^{\delay}$ is available in ls(h)\\then $\sigma(h \to^{\delay} \state) \in \TransSet$
\end{enumerate}
where $h \to^{\delay} \state$ denotes the
history obtained by adding the delay $\delay$ at the end of $h$.
\end{definition}

A strategy can return a discrete transition if its guard is satisfied and the
resulting state satisfies the destination invariant (\ref{tavail}).
To respect the alternation between timed and discrete transitions, we require that
a strategy which returns a finite delay $\delay \geq 0$ on a history returns a discrete
transition after the delay (if the run did not stop by violating an invariant)(\ref{altern}).

A controller strategy $\strat_c$ and an environment strategy $\strat_e$ can be combined
into a global strategy $\strat_{(\strat_c,\strat_e)}$ as follows.
If both players try to take a transition, we consider that the controller cannot guarantee
his transition will be taken, thus the environment chooses.
If one player decides on a discrete transition while the other decides to wait, the
discrete transition is taken.
If both players decide to wait, we wait for the smallest delay.

\begin{definition}[global strategy]
Let $\strat_c$ be a controller strategy and $\strat_e$ an environment strategy.
For all $h \in \mathcal{H}$, the global strategy $\strat_{(\strat_c,\strat_e)}$ is
defined by:
\begin{itemize}[noitemsep,topsep=0pt]
\item $\sigma_e(h) = \trans_u \in \TransSet_u \implies
    \sigma_{(\sigma_c,\sigma_e)}(h) = \trans_u$
\item $\sigma_c(h) = \trans_c \in \TransSet_c \land \sigma_e(h) = \delay\geq 0 \implies
    \sigma_{(\sigma_c,\sigma_e)}(h) = \trans_c$
\item $\sigma_c(h) = \delay\geq 0 \land \sigma_e(h) = \delay' \geq 0 \implies
    \sigma_{(\sigma_c,\sigma_e)}(h) = min (\delay, \delay')$
\end{itemize}
\end{definition}

\begin{example}
Let us look at possible strategies in location \styleloc{preparing\_coffee} of the
running example. The machine can choose \styleact{serve\_coffee} while the user
can select \styleact{ask\_sugar}. If both want to do an action, the strategy chooses
\styleact{ask\_sugar}, thus giving priority to the user. If only one of them wants
to take an action and the other waits, the action is taken. Hence, 
the machine can do \styleact{serve\_coffee} if the user is waiting. 
This is the expected behavior of a coffee machine and its user.
\end{example}

The global strategy induces a unique run, introducing null delays
between two discrete transitions to guarantee the alternation with timed transitions.

\begin{definition}[run induced by a global strategy]
Let an initial state $\state_0$ and a global strategy $\strat$ be given.
The \emph{run induced by strategy} $\strat$ is the unique
$\run_{\sigma} = s_0 s_1 s_2 \ldots$ obtained by:

\begin{itemize}[noitemsep,topsep=0pt]
    \item If $i$ is even, the next transition is a timed transition :
        \begin{itemize}
            \item If $\strat(\langle s_0, \ldots, s_i \rangle)=\trans\in\TransSet$
            a delay $0$ is added:
            $s_i\to^0 s_{i+1} \to^\trans s_{i+2}$
            \item If $\strat(\langle s_0, \ldots, s_i \rangle)$ returns a delay
            $\delta \geq 0$
            and there is a unique state $\state$ such that $s_i \temptrans \state$
            (invariant not violated), then $s_{i+1}=\state$.
            \item Otherwise, the invariant is violated and the run ends.
        \end{itemize}
    \item If $i$ is odd, the next transition is a discrete transition. By the
        properties of a strategy $\strat(\langle s_0, \ldots, s_i \rangle)$ returns
        a transition
        $\trans$ such that there is a unique state $\state$ where $\state_i \to^{\trans}
        \state$. Then, $s_{i+1} = \state$.
        \end{itemize}
\end{definition}

\begin{definition}[winning strategy]
    A controller strategy $\sigma_c$ is said to be \emph{winning} from a state $\state
    \in \StateSpace$ w.r.t.\ a reachability objective $\TargetSet$ if and only if all runs starting
    in $\state$ and \emph{adhering to} $\sigma_c$ are winning
    w.r.t.\ the objective. State $\state$ is said to be winning if there exists a winning strategy from it.
    
    Run $\run$ is \emph{adhering to} a controller strategy $\sigma_c$
    if there exists an environment strategy $\sigma_e$, such that $\run = \run_{\sigma_{(\sigma_c,\sigma_e)}}$.

\end{definition}

The question we now aim to answer is: Given a Parametric Timed Game $\game$ and a
Reachability Objective $\TargetSet$, is there a winning controller strategy from the initial state?
The question depends on the value of the parameters. So, more precisely, we are interested in the question:
\emph{For which parameter valuations is the corresponding initial state winning?}

\section{Solving the Game}
\label{sec:solve}
In this section, we introduce necessary elements for solving a game.
We first describe the symbolic state space on which the algorithm operates.
Then we characterize the set of winning states as a nested fixed point.

\subsection{Parametric Zone Graph}

Since clock valuations assign real numbers, the timed transition system of a PTG has an uncountable number of states. 
Zones (cf.\ \cref{def:zone}) are a practical tool to regroup these states in more manageable sets.
Recall that zones (like guards and invariants) are conjunctions of simple constraints on valuations,
and can be viewed as sets of valuations.
Our algorithms operate on symbolic states $\SymbState = (\styleloc{\loc},\zone)$, which
consist of a location and a zone. We require that $\zone \subseteq Inv(\styleloc{\loc})$.
For instance, the set of initial states of a PTG $\SymbState_0$ (cf. \cref{def:run}) can be described
by the symbolic state 
$(\styleloc{\loc_0}, \Inv(\styleloc{\loc_0})\wedge \bigwedge_{\styleclock{x}\in X}
\styleclock{x} = 0)$.

In the notation, we identify a symbolic state $(\styleloc{\loc},\zone)$ with
its semantics as the set of concrete states: $\{(l,v) \mid v \vDash Z\}\subseteq\StateSpace$.
We will write $\SymbState.\loc$ to denote the (common) location of a symbolic state $\SymbState$.
Zones are closed under the following operations, which we extend to symbolic states:
\begin{itemize}
    \item Intersection between sets
    \item Temporal successors: $\TempSucc{\SymbState} = \{ \state' \in \StateSpace \mid \exists \state \in \SymbState, \state \temptrans \state' \} $ 
    \item Temporal predecessors: $\TempPred{\SymbState} = \{ \state' \in \StateSpace \mid \exists \state \in \SymbState, \state' \temptrans \state \}$
    \item Discrete successors: $Succ(\trans, \SymbState) = \{ \state' \in \StateSpace \mid \exists \state \in \SymbState, \state \to^{\trans} \state' \} $
    \item Discrete predecessors: $Pred(\trans, \SymbState) = \{ \state' \in \StateSpace
    \mid \exists \state \in \SymbState, \state' \to^{\trans} \state \} $
    \item Projection onto parameters: $\ProjectParam{\xi} = 
    \{ \val_\ParamSet \mid \exists \val_\ClockSet, \loc, \, (\loc, (\val_\ClockSet,\val_\ParamSet))\in\xi \}$
\end{itemize}

\noindent
These operations can be implemented by standard operations on convex 
polyhedra~\cite{DBLP:journals/scp/BagnaraHZ08}.
We also use union, set complement and set difference, which can
return non-convex shapes. These are represented as unions of zones,
still denoting sets of concrete states. All previous operations are extended to unions of zones.

Our algorithms operate on the Parametric Zone Graph (PZG). The PZG of a PTG is not guaranteed to be finite, so our algorithms
are in fact semi-algorithms.

\begin{definition}[Parametric Zone Graph]
    Given a PTA of the form $G=(\LocSet , \ClockSet, \ParamSet, \LabelSet, \TransSet_c,
    \TransSet_u, \styleloc{\loc_0}, \Inv)$, 
    its \emph{Parametric Zone Graph} is defined as the tuple $(\SymbStateSet, \TempSucc{\SymbState_0},
    \Rightarrow^\trans_c,\Rightarrow^\trans_u)$, 
    where $\SymbStateSet \subseteq 2^\StateSpace$;
    $\forall \SymbState,\SymbState'\in \SymbStateSet$ we have $\SymbState \Rightarrow^t_c \SymbState'$ if $\SymbState' = \TempSucc{Succ(\trans,\SymbState)}$ and $\trans\in T_c$; and
    $\SymbState \Rightarrow^t_u \SymbState'$ if $\SymbState' = \TempSucc{Succ(\trans,\SymbState)}$ and $\trans\in T_u$.
\end{definition}

\subsection{Alternating Fixed Point Property}

The algorithm works by alternating between exploring new states and back-propagating
winning-state information from discovered winning states, starting from target states.
The exploration relies on a fixed point property of the set $Reach(\SymbState_0)$,
defined as all symbolic states in some run from an initial state in $\SymbState_0$:
\begin{lemma}[from \cite{DBLP:conf/wodes/JovanovicFLR12}]
    $Reach(\SymbState_0)$ is the smallest set $\StateSet$ containing $\TempSucc{\SymbState_0}$
    such that $\forall \trans \in T,\; \TempSucc{Succ(\trans,\StateSet)} \subseteq \StateSet$.
\end{lemma}

Similarly, the set of winning states $\Winning(\TargetSet)$ of the PTG with reachability objective $\TargetSet$ can be computed as a fixed point. Intuitively, we can win the game if we can take a temporal transition (without being diverted by an uncontrollable action leading us to a non-winning state)
to a state that is either directly winning, or has a controllable transition to a winning state. 
We formalize this with three operators on sets of states.
Let $\Winning$ be the set of winning states of the game.

We call $\WinningMoves(\StateSet) = \{ \state \in \StateSpace \mid \exists \trans \in T_c,  \state' \in \StateSet, \; \state \to^{\trans} \state' \}$ the set of states that have access to a controllable action leading to $\StateSet$.
When applied to $\Winning$, it gives us the states with a controllable action to reach $\Winning$, from which we have a winning strategy.
$\WinningMoves(\StateSet)$ is increasing in $\StateSet$.
It can be computed using the previous operators as 
$\WinningMoves(\StateSet) = \bigcup_{\trans_c \in \TransSet_c} Pred(t_c, \StateSet)\enspace.$

We call $\Uncontrollable(\StateSet) = \{ \state \in \StateSpace \mid \exists \trans \in T_u,  \state' \notin \StateSet, \; \state \to^{\trans} \state' \}$ the set of states where an uncontrollable action leads to a state outside $\StateSet$.
When applied to \Winning, it gives us the states where the environment can derail us into a state outside $\Winning$, from which we have no winning strategy. $\Uncontrollable(\StateSet)$ is decreasing in $\StateSet$.
It can be computed from the operators from the previous subsection by 
$\Uncontrollable(\StateSet) = \bigcup_{\trans_u \in \TransSet_u} Pred(t_u, \StateSpace \backslash \StateSet)\enspace.$

Finally, we call $\SafePred(\StateSet_1, \StateSet_2)$ the set of states that
can reach $\StateSet_1$ by a temporal transition while avoiding $\StateSet_2$.
Since it aims to be applied to reach winning moves while avoiding uncontrollable actions, if a state is in the intersection of $\StateSet_1$ and $\StateSet_2$, priority is given to the environment and the state is not considered safe. 
$\SafePred(\StateSet_1,\StateSet_2) = \{ \state \in \StateSpace \mid \exists \state' \in \StateSet_1, \state \temptrans \state' \land \Cover( \state\temptrans \state') \cap \StateSet_2 = \emptyset \}$.
$\SafePred(\StateSet_1, \StateSet_2)$ is increasing in $\StateSet_1$ and decreasing in $\StateSet_2$.

Thanks to the work of Cassez et al.~\cite{OTF-TG}, $\SafePred$
can be computed between zones using the precedent operations and extended to union of zones:

\begin{lemma}[from \cite{OTF-TG} for TG, \cite{Classic-PTG} for PTG]
$$\SafePred( \StateSet_1, \StateSet_2) = (\TempPred{\StateSet_1} \backslash \TempPred{\StateSet_2}) \cup  \TempPred{( (\StateSet_1 \cap (\TempPred{\StateSet_2})) \backslash \StateSet_2 )}$$
    $$ \SafePred( \underset{i}{\bigcup} \StateSet_{1i}, \underset{j}{\bigcup} \StateSet_{2j} ) = \underset{i}{\bigcup} \big( \TempPred{\StateSet_{1i}} \cap \underset{j}{\bigcap} \SafePred( \StateSet_{1i} \big), \StateSet_{2j}) $$
\end{lemma}

We can now formulate the fixed point property followed by $\Winning$.

\begin{restatable}{lemma}{lemfixpoint}
\label{lem:fixpoint}
    $\Winning(\TargetSet)$ is the smallest set $\StateSet$ containing $\TargetSet$
    such that \\  $\SafePred( \StateSet \cup \WinningMoves(\StateSet) , \Uncontrollable
    (\StateSet)) \subseteq \StateSet\enspace.$
\end{restatable}

\proof{See \cref{WinningProp}.}

\section{Algorithm and Correctness}
\label{sec:algo}

\begin{algorithm}
  \caption{For PTG $\game = (\LocSet, \ClockSet, \ParamSet, \LabelSet, \TransSet_c,
  \TransSet_u, \styleloc{\loc_0}, \Inv)$ and reachability objective $\TargetSet$,
  returns the set of all parameter valuations that win the game.\label{alg:main}}
  \begin{algorithmic}[1]
    \State $\Explored, \WaitingUpdate, \WaitingExplore \gets \emptyset, \emptyset, \{\TempSucc{\SymbState_0}\}$
    \Comment{Symbolic state sets}
    \State $\mathit{Win} := \{\}$
    \Comment {Map from symbolic states to unions of zones}
    \State $\Depends := \{\}$
    \Comment {Map from symbolic states to sets of symbolic states}
    \State $\WinningParam := False$
    \Statex
    \Function{solvePTG}{}
    \While{$\lnot$\Call{terminate}{ \hspace*{-1mm}}}
    \State Choose either \Call{explore}{ \hspace*{-1mm}} or  \Call{update}{ \hspace*{-1mm}}
    \EndWhile
    \State \Return $\WinningParam$
    \EndFunction
    \Statex

    \Procedure{explore}{}
    \State $\SymbState \gets extract(\WaitingExplore)$
    \For{\trans transition from \SymbState :}
    \State $\SymbState' := \TempSucc{Succ(\trans, \SymbState)}$
    \State $\Depends[\SymbState'] \gets \Depends[\SymbState'] \cup \{\SymbState\}$
    \If{$\SymbState'$ not in $\Explored$}
    \State $\WaitingExplore \gets \WaitingExplore \cup \{\SymbState'\}$
    \EndIf
    \EndFor
    \If{$\SymbState.\loc \in \TargetSet$}
    \State $\mathit{Win}[\SymbState] \gets \SymbState$
    \State $\WaitingUpdate \gets \WaitingUpdate \cup \Depends[\SymbState]$
    \EndIf
    \State $\WaitingUpdate \gets \WaitingUpdate \cup \{\SymbState\}$
    \State $\Explored \gets \Explored \cup \{\SymbState\}$
    \EndProcedure
    \Statex
    \Procedure{update}{}
    \State $\SymbState \gets extract(\WaitingUpdate)$
    \State $Uncontrollable \gets \hspace{-1em}\underset{\{(\SymbState', t) \mid  \SymbState
    \Rightarrow^{t}_u\SymbState'\}}{\bigcup}\hspace{-1em} Pred (t, \SymbState' \setminus
    \mathit{Win}[\SymbState'])$
    \State $\WinningMoves \gets \hspace{-1em}\underset{\{(\SymbState', t) \mid  \SymbState
    \Rightarrow^{t}_c\SymbState'\}}{\bigcup}\hspace{-1em} Pred (t, \mathit{Win}[\SymbState'])$
    \State $\NewWin := \SafePred( Win[\SymbState] \cup \WinningMoves, \: Uncontrollable) \:  \cap \: \SymbState$
    \If{$\NewWin \not\subseteq \mathit{Win}[\SymbState]$}
    \State $\WaitingUpdate \gets \WaitingUpdate \cup \Depends[\SymbState]$
    \State $\mathit{Win}[\SymbState] \gets  \mathit{Win}[\SymbState] \cup \NewWin$
    \EndIf
    \State $\WinningParam \gets \ProjectParam{(\mathit{Win}[\SymbState] \cap \SymbState_0)}$ \label{line:report}
    \EndProcedure
    \Statex
    \Function{terminate}{}
    \State \Return $\WaitingExplore = \emptyset \land \WaitingUpdate = \emptyset$
    \EndFunction
  \end{algorithmic}
\end{algorithm}

We can now introduce the algorithm for parameter synthesis for PTG.
\Cref{alg:main} explores the state space and creates a map of symbolic states
connected by a discrete transition through the operation $\TempSucc{Succ(\trans_s, \_ )}$. 
Simultaneously, any newly found winning states in a symbolic state $\SymbState$, starting from the target locations, are propagated by marking the predecessors of $\SymbState$ for an update. 
To update a symbolic state $\SymbState$, we compute $\SafePred({\mathit{Win} \cup
\WinningMoves(\mathit{Win})},$ $ \Uncontrollable(\mathit{Win}))$ within $\SymbState$
and add the result to $\mathit{Win}[\SymbState]$. If new winning states are found, we mark $\SymbState$ predecessors for an update.

\smallskip

The algorithm is non-deterministic: it does not describe how we choose between explore
and update, and which symbolic state in the waiting lists to explore or update. These choices are left abstract on purpose, as optimization opportunities.
A fair strategy would be to join \WaitingExplore\ and \WaitingUpdate\ in a single queue, whose head
determines which operation to apply next. In our implementation,
we prioritized back-propagation from \WaitingUpdate.

\subsection{Invariants and Correctness}

Recall that the algorithm works on a zone graph. We are looking for subsets of winning states within symbolic states. The same state may appear in different symbolic states and may not have the same status in each instance. Therefore, the set $\Win$, the winning states found by the algorithm so far, and $\Winning$, the set of all winning states, also take into account the symbolic state considered. Formally,  
 $\Winning$ consists of all pairs $(\SymbState, \state)$ where $\state$ is a winning state contained in the symbolic state $\SymbState$, and
$\Win = \hspace{-2ex}\underset{\SymbState \in \Explored}{\bigcup}\hspace{-2ex} \{\SymbState\} \times Win[\SymbState]$.

\begin{restatable}{theorem}{thminv} \label{thm:inv}
\label{thm:inv}
  These invariants hold during the execution of the algorithm:
  \begin{enumerate}[topsep=0pt]
      \item $ \TempSucc{\SymbState_0}\in\Explored$.
    \item $\forall \SymbState \in \Explored, \trans\in \TransSet, \SymbState'$, if $\SymbState \Rightarrow^{\trans}\SymbState'$, 
    then $\SymbState' \in {\WaitingExplore \cup \Explored}$ \label{inv:fp_expl}
    \item $\forall \SymbState \in \Explored$, if $\SymbState.l \in \TargetSet$, $\{\SymbState\} \times \SymbState \subseteq \Win $.
    \item $\Win \subseteq \Winning$. \label{inv:correctness}
    \item $\forall \SymbState \in \Explored$, we have either $\SymbState \in \WaitingUpdate$ or
    $\SafePred( \Win \cup \WinningMoves(\Win) , \: \Uncontrollable(\Win) ) \cap (\{ \SymbState\} \times \SymbState) \subseteq \Win$. \label{inv:fp_upd}
  \end{enumerate}
\end{restatable}

\begin{proof}
See \cref{ProofInv}.
\end{proof}

Invariant 4 guarantees that even if the algorithm times out the winning states found
by the algorithm are indeed winning. 
Furthermore, if the algorithm terminates and the waiting lists are empty, we can apply the fixed point properties of 
$Reach(\SymbState_0)$ and $\Winning$, and $\Win$ corresponds exactly to $\Winning$ 
over the explored symbolic states that cover $Reach(\SymbState_0)$.

\begin{restatable}{theorem}{algcorrect}
  \Cref{alg:main} is correct (when it terminates).
\end{restatable}

\begin{proof}
  See \cref{ProofInv}.
\end{proof}

\begin{example}
For the coffee machine, the PZG is only finite after applying inclusion subsumption
(\cref{sec:optimizations}). However, even on this finite PZG, \cref{alg:main} does not terminate,
but keeps reporting solutions at \cref{line:report}.
In fact, it produces increasingly more general solutions, including 
$n\, \styleparam{p_2} \geq  \styleparam{p_1}$ (for any $n>0$).
If we bound these parameters in the initial specification, for intance 
$\styleparam{p_2}\geq 1 \wedge \styleparam{p_1}\leq 5$, our
algorithm synthesizes the extra constraints 
$\styleparam{p_1}+\styleparam{p_2}\leq\styleparam{p_4} \wedge \styleparam{p_3}<\styleparam{p_4}$, as expected.
\end{example}

\begin{theorem}
  Provided the waiting lists are treated fairly, any explored winning state is discovered as winning by the algorithm eventually.
\end{theorem}

\begin{proof}[sketch]
For a classical TG, we can represent the underlying TA as a finite classical automaton
(\eg{} the region graph). On this automaton, we can define a (finite) turn-based reachability
game equivalent to the initial TG. Hence, we can use the notion of \emph{discrete distance to target} in a reachability game, corresponding to the smallest number of discrete transitions in which a controller can ensure to reach a target. This is equivalent to solving a \emph{Min-Cost Reachability game} as studied in \cite{DBLP:conf/concur/BrihayeGHM15} where delay transitions have weight $0$ and discrete transitions have weight $1$. The game graph is finite and the weights non-negative, so the discrete distance to target of a winning state is positive and finite.

While the same construction is not necessarily finite in a PTG, any state of a PTG is a state $(\loc, (\val_\ParamSet, \val_\ClockSet))$ of the TG, where all parametric linear terms in guards have been replaced by their valuation through $\val_\ParamSet$. Therefore, this result extends to winning states of a PTG.

Let $s$ be an explored winning state of the PTG and $n$ its distance to target. We only need to explore states reachable in $n$ discrete transitions from $s$. By invariant \ref{inv:fp_expl} from \cref{thm:inv}, when all states reachable in $k$ discrete transitions are explored, all states reachable in $k+1$ discrete transitions are either already explored or in the exploration waiting list. Assuming fairness of the waiting lists, at some time they have all been explored. Therefore, at some time, all states reachable from $s$ in $n$ discrete steps have been explored.

When all states reachable in $n$ discrete steps have been explored, all target states within are discovered. Those are states with distance to target $0$. For $0\leq k<n$, when all winning states reachable in less than $n-k$ discrete transitions from $s$ and with a distance to target less than $k$ are discovered, then all winning states reachable in less than $n-(k+1)$ discrete transitions from $s$ and with a distance to target less than $k+1$ are discoverable by update. Using the invariant (\ref{inv:fp_upd}) of \cref{thm:inv}, those states are either already discovered as winning or they are in the update waiting list. Assuming fairness of the waiting lists, at some time  they have all been discovered winning. Applying this recurrence until $k=n$, we get that there is a time where $s$ is discovered winning.
\end{proof}

We can guarantee: (1) All winning parameter valuations reported in \cref{line:report} are
correct, since the algorithm satisfies the invariants of \cref{thm:inv}.
(2) Every winning parameter valuation will eventually be reported, provided the 
waiting lists are treated fairly.
Hence, \cref{alg:main} is ``sound and complete in the limit'' \cite{DBLP:conf/tacas/AndreAPP21}.

\section{Optimizations}
\label{sec:optimizations}

\begin{algorithm}[t]
  \caption{Adding optimizations to the explore procedure}
  \label{alg:explore_opt}
  \begin{algorithmic}[1]
    \Procedure{explore}{}
    \State $\SymbState \gets extract(\WaitingExplore)$
    \BoxedIf[draw=ColorCumuPrune, fill=ColorCumuPrune!20!white]{ 
      $\ProjectParam{\SymbState} \subseteq \WinningParam$} \label{cumulativestart}
      \Comment{Cumulative Pruning}\SetBoxEast
      \State \Return \label{cumulativeend}\EndBox
      \EndIf
    \If{$\SymbState.\loc \in \TargetSet$}
    \State $\mathit{Win}[\SymbState] \gets \SymbState$
    \State $\WaitingUpdate \gets \WaitingUpdate \cup \Depends[\SymbState]$
    \EndIf
    
      \BoxedIf[draw=ColorLosingProp, fill=ColorLosingProp!20!white]{$\ControllableDeadlock(\SymbState) \; \land \; \mathit{Win}[\SymbState] \neq
      \SymbState $}\label{deadlockstart}
      \Comment{Losing state propagation}\SetBoxEast
      \State $\WaitingUpdate_L \gets \WaitingUpdate_L \cup \Depends[\SymbState]$\label{deadlockend}\EndBox
      \EndIf

      \BoxedIf[draw=ColorCovPrune, fill=ColorCovPrune!20!white]{$\mathit{Win}[\SymbState] = \SymbState\; \lor \; \ControllableDeadlock(\SymbState)$}
      \label{covprunestart}
      \Comment{Coverage Pruning}\SetBoxEast
      \State \Return \label{covpruneend}\EndBox
      \EndIf
    \For{\trans transition from \SymbState :}
    \State $\SymbState' := \TempSucc{Succ(\trans, \SymbState)}$
      \BoxedIf[draw=ColorInc, fill=ColorInc!20!white]{$\exists \SymbState'' \in \Explored: \SymbState'\subseteq \SymbState''$}\label{inclusioncheckstart}
      \Comment{Inclusion check}\SetBoxEast
      \State $\Depends[\SymbState''] \gets \Depends[\SymbState''] \cup \{\SymbState\}$
    \Else
    \State $\Depends[\SymbState'] \gets \Depends[\SymbState'] \cup \{\SymbState\}$
    \State $\WaitingExplore \gets \WaitingExplore \cup \{\SymbState'\}$\label{inclusioncheckend}\EndBox
    \EndIf
    \EndFor
    \State $\WaitingUpdate_W \gets \WaitingUpdate_W \cup \{\SymbState\}$\label{winningpartition}
    \BoxedState[draw=ColorLosingProp, fill=ColorLosingProp!20!white] $\WaitingUpdate_L \gets \WaitingUpdate_L \cup \{\SymbState\}$\label{deadlockpart}
      \Comment{Losing state propagation}\EndBox
    \State $\Explored \gets \Explored \cup \{\SymbState\}$
    \EndProcedure
  \end{algorithmic}
\end{algorithm}

We present four optimizations to the algorithm presented in Section~\ref{sec:algo}.
All of them adapt optimizations from previous works, three of them (coverage pruning,
inclusion checking and losing state propagation) from Cassez et al.~\cite{OTF-TG}
and one of them (cumulative pruning) from Andr{\'{e}} et al.~\cite{DBLP:conf/tacas/AndreAPP21}.
We start by updating the exploration procedure to include the optimizations, as shown
in~\cref{alg:explore_opt}.
\subsection{Pruning}
First, we present some pruning techniques, as these only require slight modifications in the exploration procedure. To this end, we introduce the notion of a \textit{controller deadlock} state. A state is a \textit{controller deadlock} state if it has no controllable transitions. We define it as the following predicate on symbolic states:
\[\ControllableDeadlock(\SymbState) = \forall \trans, \SymbState': \mbox{ if }\SymbState \Rightarrow^{t} \SymbState' \mbox{ then } t\in \TransSet_u\]
Now, we introduce the two kinds of pruning:
\begin{itemize}
  \item \textbf{Cumulative Pruning:} If the projected parameters of a zone in a new
  symbolic state are included in the current set of winning parameters, we can safely prune the successors of this state. Indeed, if the only possible parameters in the zone already are determined to be winning, no new winning parameter can be found by exploring the successors of this state. This check can be seen in \cref{cumulativestart,cumulativeend} of \cref{alg:explore_opt}.
  \item \textbf{Coverage Pruning:} If a symbolic state is either winning or a controllable deadlock
  state, its successors can safely be pruned. Indeed, if the symbolic state is winning,
  we gain nothing from exploring further. Dually, a controller deadlock state can never become
  winning, since the controller has no action to do. This check can be seen in \cref{covprunestart,covpruneend}
  of \cref{alg:explore_opt}.
\end{itemize}

\subsection{Inclusion checking}
Originally, checking if a symbolic state $\SymbState'$ has been explored already is
done by checking if $\SymbState'\in \Explored$. The optimization by inclusion checking
instead checks if $\exists \SymbState'' \in \Explored : \SymbState'\subseteq \SymbState''$.
If this is the case, the newly discovered symbolic state can safely be discarded since
its superset has already been explored. Of course, the new dependency that $\SymbState$
depends on $\SymbState''$ still must be added. This optimization is done in the
exploration procedure (\crefrange{inclusioncheckstart}{inclusioncheckend} of \cref{alg:explore_opt}).
\subsection{Losing state propagation}
Losing state propagation is inspired by Cassez et al. \cite{OTF-TG} for TG. The idea is that instead of only discovering and propagating winning states, we will now also do the same for losing states, starting from controller deadlock states. A map $Lose$ will maintain the currently known losing states for a given symbolic state. Thus, each symbolic state $\SymbState$ can now be partitioned into three:
\begin{itemize}
  \item Winning: $Win[\SymbState]$,
  \item Losing: $Lose[\SymbState]$
  \item Unknown: $\SymbState \setminus (Win[\SymbState]\cup Lose[\SymbState])$.
\end{itemize}
To initially mark a state as losing, we use the controller deadlock predicate again while also making sure that the state is not winning, as shown
in \cref{alg:explore_opt}, \cref{deadlockstart,deadlockend}.
On \cref{winningpartition,deadlockpart}, we partition the $\WaitingUpdate$ list into
two lists for propagating winning and losing states respectively.

While pruning and inclusion checking only required the modification of the exploration procedure, the propagation of losing states influences all of the procedures of the original algorithm. We go through them now.

\paragraph{Update procedure.}
We create a new procedure for updating losing states, which can be seen in \cref{alg:update_opt}.
As the dual of the original update procedure, it is almost identical. Instead of $\Uncontrollable$,
we compute $\Controllable$, \ie{} the union of zones where the controller can lead
to a non-losing state. Similarly, instead of $\WinningMoves$, we compute $\LosingMoves$ which is the set of states where the environment can lead to a losing state. We then compute $\NewLosing$ which is the set of states where the environment can lead us to a losing state while avoiding states where \emph{only} controllable transitions are enabled ($\Controllable
\setminus \LosingMoves$). Finally, we update $Lose[\SymbState]$ and $\WaitingUpdate_L$ accordingly.

\begin{algorithm}[t]
  \caption{Adding new update procedure for losing state propagation}
  \label{alg:update_opt}
  \begin{algorithmic}
    \Procedure{update\_l}{}
    \State $\SymbState \gets extract(\WaitingUpdate_L)$
    \State $\Controllable \gets \hspace{-1em}\underset{\{(\SymbState', t) \mid  \SymbState \Rightarrow^{t}_c\SymbState'\}}{\bigcup}\hspace{-1em} Pred (t, \SymbState' \setminus Lose[\SymbState'])$
    \State $\LosingMoves \gets \hspace{-1em}\underset{\{(\SymbState', t) \mid  \SymbState \Rightarrow^{t}_u\SymbState'\}}{\bigcup}\hspace{-1em} Pred (t, Lose[\SymbState'])$
    \State $\NewLosing := \SafePred(Lose[\SymbState] \cup \LosingMoves, \: \Controllable \setminus \LosingMoves) \:  \cap \: \SymbState$
    \If{$Lose[\SymbState] \subsetneq \NewLosing$}
    \State $\WaitingUpdate_L \gets \WaitingUpdate_L \cup \Depends[\SymbState]$; $Lose[\SymbState] \gets \NewLosing$
    \EndIf
    \EndProcedure
  \end{algorithmic}
\end{algorithm}

\paragraph{Terminate function.}
The terminate function is modified to allow for early termination if all possible
information is already known, \ie{} $\TempSucc{\SymbState_0} \setminus (Win[\TempSucc
{\SymbState_0}]\cup Lose[\TempSucc{\SymbState_0}]) = \emptyset$. Indeed, if for all
valuations we have determined that we either win or lose, the algorithm can safely
terminate. This is shown in \cref{alg:terminate_opt}.

\begin{algorithm}[t]
  \caption{New \textsc{Terminate} with early termination if initial zone is covered}
  \label{alg:terminate_opt}
  \begin{algorithmic}
    \Function{terminate}{}
    \State $isEmpty \gets \WaitingExplore = \emptyset \land \WaitingUpdate = \emptyset$
    \State $initialZoneCovered \gets (\TempSucc{\SymbState_0}) \subseteq (Win[\TempSucc{\SymbState_0}] \cup Lose[\TempSucc{\SymbState_0}])$
    \State \Return $isEmpty \; \lor \; initialZoneCovered$
    \EndFunction
  \end{algorithmic}
\end{algorithm}

The final algorithm is then modified to include the new procedures and data structures introduced. As a result, in the main loop we now have to choose between three waiting lists instead of two: $\WaitingExplore$, $\WaitingUpdate_W$ and $\WaitingUpdate_L$.

\section{Implementation and Experimental Evaluation}
\label{sec:expe}

To evaluate the termination behavior and efficiency of the semi-algorithm and the optimizations,
we implement them in the \imitator toolset and measure the performance on some realistic
case studies.

\subsection{Implementation}
We have implemented our proposed algorithm and optimizations
in the \imitator model checker \cite{DBLP:conf/cav/Andre21}, which features a wide repertoire of synthesis algorithms for PTA. We have extended its input language to PTG and added our PTG parameter synthesis algorithm, including the optimizations described in \cref{sec:optimizations}. The source code (in OCaml) is available on github\footnote{\url{https://github.com/imitator-model-checker/imitator}, branch: develop}.

In \imitator, the user specifies a model consisting of parameters, clocks and a network
of parametric timed automata. The user can analyse the model using an
analysis or synthesis query. \imitator selects the corresponding algorithm to use, after which it outputs the result of the query.

Our extension enables the user to specify edges in a PTA as (un)controllable, effectively turning it into a PTG. Along with this we add a new property \texttt{Win} and a corresponding algorithm \texttt{AlgoPTG}. In order to synthesize parameters for a PTG one must use
\texttt{property := \#synth Win(state\_predicate)},
using a predicate to define which states are winning. Usually, this predicate is simply
\texttt{accepting}, meaning that any state in an accepting location of the PTG is winning.

In \cref{alg:main} we left the choice between exploration and back-propagation to
be non-deterministic. In the implementation we choose to prioritize back-propagation
over exploration whenever possible (\ie{} when $\WaitingUpdate$ is non-empty), otherwise we explore. 
This seems to yield the fastest results in practice. 

\subsection{Experiment Design}
We selected two large case studies, one PTA and one TG, and extended them to PTGs by adding (un)controllable actions, and clock parameters, respectively.
An artifact containing instructions to run all the experiments is available online~\cite{ARTIFACT}.

\paragraph{Production Cell.}
This case study \cite{ProdCellCaseStudy2} has two conveyor belts (1 / 2), a robot with two arms (A / B) and a press. Plates arrive at conveyor belt 1 and are taken to the press by robot arm A, where they are processed for some time. Robot arm B takes processed plates and removes them through conveyor belt 2.

We model systems with 1--4 plates in \imitator. In the goal location, every plate made it safely to conveyor belt 2. If two plates collide before they are picked up by arm A, the game is lost immediately.
We assume that the rotation speed of the robot arm, the speed of the conveyor belt and the time to press are known constants. The aim is to synthesize a parameter \texttt{MINWAIT}, the minimum time interval between two plates arriving at the conveyor belt. The maximum time interval between two plates is fixed by
an additional constant \texttt{MAXWAIT}.

Our PTG model is largely inspired by the TG model of Cassez et al.~\cite{OTF-TG}.
Besides adding parameters, we check for collisions between plates rather than defining a maximum waiting time frame.
For 2--4 plates, we create a winning and a losing configuration of the constants; for 1 plate a collision is not possible. The losing configurations are created by setting \texttt{MAXWAIT} too small, which will deadlock the system for any value of \texttt{MINWAIT}. 

The \imitator model for the 1-plate configuration can be seen in Appendix~\ref{app:prodcell}.

\paragraph{Bounded Retransmission Protocol.}
The BRP provides reliable communication over an unreliable channel. We create a PTG
from a PTA model of the BRP~\cite{DBLP:conf/tap/AndreMP21}, in turn based on a TA
model \cite{BRP}, by making message loss uncontrollable.

In the BRP, a sender sends message frames to a receiver, tagged with an alternating bit, through a lossy channel. The receiver acknowledges all frames. If the sender does not receive an acknowledgement in time, it retransmits the message at most $k$ times, after which the sender gives up. The goal location indicates the successful transmission of the message, or the abort by the sender. 

\paragraph{Experimental Setup.}
All experiments were run on a single core of a computer with an Intel Core i5-10400F CPU @ 2.90GHz with 16GB of RAM running Ubuntu 20.04.6 LTS. For each implementation (basic, inclusion checking, cumulative pruning, coverage pruning, losing state propagation) we run the experiments 5 times and report the average time and state space size. A timeout of 2 hours is used.

\subsection{Experimental Results}
\begin{table}
    \begin{center}
    \footnotesize
    \caption{\label{experiment-table}Experimental results for different optimizations: inclusion check (inc), cumulative pruning (cm), coverage pruning (cv), losing state propagation (lp). Running time in seconds (s) and number of symbolic states (size). Green indicates the best results.}

        \begin{tabular}{|c|c||c|c||c|c||c|c||c|c|}
            \hline
            \multicolumn{2}{|c||}{\multirow{2}{*}{}} & \multicolumn{2}{c||}{inc} & \multicolumn{2}{c||}{inc+cm} & \multicolumn{2}{c||}{inc+cm+cv} & \multicolumn{2}{c|}{inc+cm+cv+lp} \\
            \cline{3-10}
            \multicolumn{2}{|c||}{\multirow{2}{*}{}} & Time & Size & Time & Size & Time & Size & Time & Size\\
            \hline
            \rowcolor{lightgray} \multicolumn{10}{|c|}{} \\
            \rowcolor{lightgray}
                \multicolumn{1}{|c}{\centering
                    \parbox{0.4cm}{\multirow{-2}{*}{\rotatebox{90}{\scriptsize
                    \mbox{\hspace{-2.5ex}plates}}}}}
                & \multicolumn{9}{c|}{\multirow{-2}{*}{Production Cell}} \\
            \hline
            1 & Win & 0.06s & 86 & 0.06s & 86 & 0.06s & 86 & 0.08s & 86 \\
            \hline
            \multirow{2}{*}{2} & Win & 7.19s & 746 & 7.56s & 746 & \cellcolor{\ColorBestResult}6.60s &
            \cellcolor{\ColorBestResult}701 & 7.22s & 701 \\
            \cline{2-10} 
            & Lose & \cellcolor{\ColorBestResult}1.43s& \cellcolor{\ColorBestResult}439 & 1.44s & 439 & 2.03s & 517 & 2.17s & 517 \\
            \hline
            \multirow{2}{*}{3} & Win & 36.7s & 1900 & 37.3s & 1900 & \cellcolor{\ColorBestResult}24.0s & \cellcolor{\ColorBestResult}1539 & 34.2s & 1539 \\
            \cline{2-10} 
            & Lose & 13.4s & 1372 & 13.9s & 1372 & \cellcolor{\ColorBestResult}9.53s & \cellcolor{\ColorBestResult}1251 & 14.2s & 1251 \\
            \hline
            \multirow{2}{*}{4} & Win & 4903s & 10755 & 4750s & 10755 & \cellcolor{\ColorBestResult}2394s & \cellcolor{\ColorBestResult}9350 & 3522s & 9350 \\
            \cline{2-10} 
            & Lose & 34.8s & 2605 & 35.6s & 2605 & \cellcolor{\ColorBestResult}21.6s & \cellcolor{\ColorBestResult}2372 & 153s & 2372 \\
            \hline
            \rowcolor{lightgray} \multicolumn{10}{|c|}{Bounded Retransmission Protocol}\\
            \hline
            \multicolumn{2}{|c||}{} & 34.3s & 1042 & 32.2s & 1042 & \cellcolor{\ColorBestResult}7.1s & \cellcolor{\ColorBestResult}612 & 7.5s & 612 \\
            \hline
        \end{tabular}
    
    \end{center}
\end{table}
We present the results of the experiments in \cref{experiment-table}. 
We do not include the runs without optimizations as they all timed out. This indicates that inclusion checking is the most vital optimization and should always be enabled. 

Indicated in green cells are the best results for each row. We can clearly see that
coverage pruning has the biggest effect of all the optimizations in our experiments. Losing state propagation seems to not provide much benefit in these experiments, as the overhead overshadows any positive effect it might have had.

\section{Conclusion}
\label{sec:concl}
We provide the first implementation of parameter synthesis for Parametric Timed Games
with reachability objectives,
based on an on-the-fly algorithm~\cite{OTF-TG,DBLP:conf/atva/JovanovicLR13}.
It appears that without additional pruning heuristics, the algorithm
cannot handle the case studies, Bounded Retransmission Protocol and Production Cell.
Inclusion subsumption is a minimal requirement to achieve any 
result.

Contrary to previous algorithms for PTA~\cite{DBLP:conf/tacas/AndreAPP21} and TG~\cite{OTF-TG}, the
parameter synthesis algorithm does not terminate, even if the parametric zone graph
is finite. But we found that in the limit all parameter values will be enumerated.

We added additional pruning techniques (coverage pruning and cumulative 
pruning) to further reduce the search space. These techniques generally increased
the speed. We also experimented with propagating losing states, but in our examples
the overhead of checking and propagating losing states was not compensated by any
pruning potential. Future work could study under which circumstances the propagation
of losing states could be beneficial, but also strengthen the detection of (partially)
losing states. Another venue for future work is to study other objectives, like
safety games or liveness conditions.

\paragraph{Acknowledgment.}
We thank Étienne André for
his help with integrating our algorithm in the \imitator tool set.

\bibliographystyle{plain}
\bibliography{bibliography.bib}

\newpage
\appendix

\section{Proof of the Fixed Point Property of \Winning}

\label{WinningProp}

\paragraph*{Notations.} In this part, for two histories $\hist = s_0 s_1 \ldots s_i$
and $\hist'= s_i s_{i+1} \ldots s_j$, we note $h \cdot h'$ the history $s_0 s_1 \ldots s_{i-1} s_{i+1} \ldots s_j$ if there exist $\delta_1,\delta_2 \geq 0$ such that $s_{i-1} \to^{\delta_1} s_i \to^{\delta_2} s_{i+1}$, and $s_0 s_1 \ldots s_i s_{i+1} \ldots s_j$ if not. 
We call $h$ a prefix of $h \cdot h'$ and $h'$ a suffix. Similar notations apply by replacing $h'$ with a run. We also define the cover of a history $h = \langle s_0, s_1, \ldots s_n \rangle$ to be:\\ $\Cover(h) = \underset{i\in \mathbb{N}}{\bigcup}\;  \Cover(\state_{2i} \temptrans \state_{2i+1})$.

The proof of \cref{lem:fixpoint} uses the following \cref{lem:propwinning,lem:winintersect}.

\begin{lemma}
\label{lem:propwinning}
  $\SafePred( \Winning(\TargetSet) \cup \WinningMoves(\Winning(\TargetSet)) , \Uncontrollable
  (\Winning(\TargetSet))) \subseteq \Winning(\TargetSet)\enspace.$
\end{lemma}

\begin{proof}
    Let $s \in \SafePred( \Winning(\TargetSet) \cup \WinningMoves(\Winning(\TargetSet)) , \Uncontrollable(\Winning(\TargetSet)))$.
    
    There exists $s' \in \Winning(\TargetSet) \cup \WinningMoves(\Winning(\TargetSet))$ and $\delta \geq 0$ such that $s \temptrans s'$ and $\Cover(s \temptrans s') \cap \Uncontrollable(\Winning(\TargetSet)) = \emptyset$.

    We define a controller strategy $\strat_c$ from $s$ such that $\strat_c(s) =
    \delta$. If $s'$ is not winning, it is in $\WinningMoves(\Winning)$ and thus have
    a controllable discrete transition $\trans_c$ leading to a state $s''$ such that
    $s''$
    is in $\Winning(\TargetSet)$. In this case, we set $\strat_c(s,s')=\trans_c$.

    Now, whatever the environment strategy $\strat_e$, the possible first moves of
    the run are limited to $s \to^{\delta} s'$ with $s'$ winning, $s \to^{\delta}
    s' \to^{t_c} s''$ with $s''$ winning or $s \to^{\delta_i} s_i \to^{t_u} s_i'$
    with $s_i \in \Cover(s \to^{\delta} s')$ and $t_u \in T_u$. Since $s_i\in \Cover
    (s \temptrans s')$, $s_i \notin \Uncontrollable(\Winning(\TargetSet))$ and so
    $s_i'$
    is winning.

    By definition, for every winning state $s_w$, there exists a controller strategy
    $\strat_c^{s_w}$ that is winning from $s_w$.
    We modify the strategy $\strat_c$ so that for all the histories $h_{\leq s_w}$
    described above ending in a winning state $s_w$, the controller continues the run with the strategy $\strat^{s_w}_{c}$.
    
    For every history $\hist$ starting in $s_w$, we modify strategy $\strat_c$ so
    that $\strat_c(h_{\leq s_w} \cdot \hist) = \strat^{s_w}_c(\hist)$.
    If $\strat_c^{s_w}(\langle s_w \rangle ) = \delta_2 \geq 0$ and $\hist_{\leq s_w} = h_{< s_w} \cdot \langle s_w', s_w \rangle$ with  $s_w'$ and $s_w$ the last two states of $h_{\leq s_w} $ such that there exists $\delta_1 \geq 0$ with $s_w' \to^{\delta_1} s_w$, we modify $\strat_c(\hist_{<s_w})$ to return $\delta_1+\delta_2$.

    Let $\strat_e$ be an environment strategy, $r$ be the run resulting from $\strat_e$ and $\strat_c$ and $h_{\leq s_w}$ the prefix of $r$ ending in a winning state described above.
    We call $\strat_e^{s_w}$ the environment strategy from $s_w$ such that $\strat_e^{s_w}(h) = \strat_e (h_{\leq s_w} \cdot h)$. 
     
    We have:
    $r = \hist_{\leq s_w} \cdot r_{\geq s_w}$ where $r_{\geq s_w}$ is the run resulting from $\strat^{s_w}_{c}$ and $\strat_e^{s_w}$. Since $\strat^{s_w}_{c}$ is winning, so is $r_{\geq s_w}$. Its cover contains a state with a target location and is contained in the cover of $r$. Therefore $r$ is winning and so is $\strat_c$.
    Finally, $s$ is winning.

\end{proof}
  
\begin{lemma}
\label{lem:winintersect}
  Let $\StateSet_1, \StateSet_2$ such that for $i \in \{1,2\}$, \\$\SafePred
  ( \StateSet_i \cup \WinningMoves(\StateSet_i) , \Uncontrollable(\StateSet_i))
  \subseteq \StateSet_i\enspace.$

  Then $\SafePred( (\StateSet_1 \cap \StateSet_2) \cup \WinningMoves(\StateSet_1
  \cap \StateSet_2) , \Uncontrollable(\StateSet_1 \cap \StateSet_2)) \subseteq
  (\StateSet_1 \cap \StateSet_2)\enspace.$
\end{lemma}

\begin{proof}
  For $i \in \{ 1, 2\}$, $\StateSet_1 \cap \StateSet_2 \subseteq \StateSet_i$. By monotonicity of $\WinningMoves$ and $\Uncontrollable$, we have:
  $\WinningMoves(\StateSet_1 \cap \StateSet_2) \subseteq \WinningMoves(\StateSet_i)$ and $\Uncontrollable(\StateSet_1 \cap \StateSet_2) \supseteq \Uncontrollable(\StateSet_i)$. By monotonicity of $\SafePred$,
  $\SafePred( (\StateSet_1 \cap \StateSet_2) \cup \WinningMoves(\StateSet_1 \cap \StateSet_2) , \Uncontrollable(\StateSet_1 \cap \StateSet_2)) \subseteq $ $\SafePred( \StateSet_i \cup \WinningMoves(\StateSet_i) , \Uncontrollable(\StateSet_i)) \subseteq \StateSet_i$.
  
  Thus $\SafePred( (\StateSet_1 \cap \StateSet_2) \cup \WinningMoves(\StateSet_1 \cap \StateSet_2) , \Uncontrollable(\StateSet_1 \cap \StateSet_2)) \subseteq (\StateSet_1 \cap \StateSet_2) $.
\end{proof}

\lemfixpoint*

\begin{proof}
  
    Since the property to $\SafePred( \StateSet \cup \WinningMoves(\StateSet) , \Uncontrollable(\StateSet)) \subseteq \StateSet$ for a set of states $\StateSet$ is stable by intersection, there is a smallest set $\Winning_?$ containing the target states verifying this property.

    Let $s \notin  \SafePred( \Winning_? \cup \WinningMoves(\Winning_?) , \Uncontrollable
    (\Winning_?))$, such that $s.l\notin \TargetSet$ and let $\strat_c$ be a controller
    strategy from $s$.
    We build an environment strategy $\strat_e$ from $s$ defined recursively on the
    resulting history of the combined strategy $\strat=(\strat_c,\strat_e)$ such
    that the resulting run does not enter $\Winning_?$ and thus never reaches the
    target state.

    For $h_0 = \langle s \rangle$ where $\Cover(h_0) = \{s\}$, the location of $s$
    is not a target and $ls(h_0) = s \notin \Winning_?$.

    For $n \in \mathbb{N}$,
    Let $h_{2n} = \langle s_0, s_1, s_2 \ldots s_{2n} \rangle$ be a history such that
    there exists an environment strategy $\strat_e$ defined on the strict prefix of $h_{2n}$, $h_{2n}$ is the prefix up to rank $2n$ of the resulting run of $\strat=(\strat_c,\strat_e)$.
    $\Cover(h_{2n})$ does not contain a state with a target location, $s_{2n} \notin
    \Winning_?$. We will prove that either we can extend $\sigma_e$ on $h_{2n}$ such that the resulting run of $\strat=(\strat_c,\strat_e)$ is $s_0, s_1, s_2, \ldots s_{2n}$ and is losing, or we can extend the history to $h_{2n+2} = h_{2n} \cdot \langle s_{2n}, s_{2n+1}, s_{2n+2} \rangle$ such that 
    $\Cover(h_{2n})$ does not contain a state with a target location, $s_{2n} \notin \Winning_?$, and we can extend $\sigma_e$ on $h_{2n}$ and $h_{2n+1}$ such that $h_{2n+2}$ is the prefix up to rank $2n+2$ of the resulting run of $\strat=(\strat_c,\strat_e)$.

    If $\strat_c(h_{2n}) = \infty$, we set $\strat_e(h_{2n}) = \infty$ and the resulting
    run ends here. The run cover is $\Cover(h_{2n}) \cup \Cover(s_{2n})$. By recurrence
    hypothesis, $\forall s \in \Cover(h_{2n}), \; s.l\notin \TargetSet $. Since $s_
    {2n}$ is not winning, its location is not in $\TargetSet$ and $\forall s \in \Cover(s_{2n}), \; s.l\notin \TargetSet $ thus the run is losing.
    
    Else there is a delay $\delta \geq 0$ followed by a controllable transition $\trans_c$
    and two states $s', s''$ such that $s_{2n} \to^{\delta} s' \to^{t_c} s''$
    and if $\strat_c(h_{2n}) = t_c$ if $\delta=0$ or $\strat_c(h_{2n}) = \delta$ and
    $\strat_c(h_{2n}\cdot\langle s_{2n}, s'\rangle) = \trans_c$.
    Since $s_{2n} \notin  \SafePred( \Winning_? \cup \WinningMoves(\Winning_?) , \Uncontrollable
    (\Winning_?))$, there are two possibilities. 
    First, $s''\notin\Winning_?$. In that case, we set $\strat_e(h_{2n}) = \infty$
    and we let the transition happen. $s_{2n+1} = s'$, $s_{2n+2}= s''$, all states
    in $\Cover(s_{2n} \to^{\delta} s_{2n+1})$ are in location $s_{2n}.l \notin\TargetSet$
    and $s_{2n+2} \notin \Winning_?$. %
    Second possibility, there is a state $s_i \in \Cover{(s_{2n} \to^{\delta} s')}$,
    a delay $0 \leq \delta_i \leq \delta$, an uncontrollable transition $t_u$ and
    a state $s_i'$ such that $s_{2n} \to^{\delta_i} s_i \to^{t_u} s_i'$ and $s_i'\notin
    \Winning_?$. We set $\strat_e(h_{2n}) = \delta_i$ and $\strat_e(h_{2n}\cdot \langle
    s_{2n},s_i\rangle)= t_u$. Since $\strat_e(h_{2n}) = \delta_i \leq \delta = \strat_c
    (h_{2n})$, the next states of the run are $ s_{2n+1} = s_i$ and $ s_{2n+2} = s_i'$.
    All states in $\Cover (s_{2n} \to^{\delta} s_{2n+1})$ are in location $s_{2n}.l\notin
    \TargetSet$ and $s_{2n+2} \notin \Winning_?$.

    Thus, either the resulting run of $\strat_c$ and $\strat_e$ ends and is losing,
    or it is infinite and its cover is the union of the covers of its histories which does not contain states with a target location. Therefore, the run is losing and $s$ is not in $\Winning_?$.
    
    Finally, we get $\Winning_? \subseteq \Winning(\TargetSet)$ and since $\Winning_?$
    is the smallest set containing the target states where $\SafePred( \StateSet
    \cup \WinningMoves(\StateSet) , \Uncontrollable(\StateSet)) \subseteq
    \StateSet$ holds
    and since $\Winning(\TargetSet) \subseteq \Winning_?$, we obtain
    $\Winning(\TargetSet) = \Winning_?$.

\end{proof}

\section{Proof of the Correctness of the Algorithm}
\label{ProofInv}

Recall that $\Win = \underset{\SymbState \in \Explored}{\bigcup} \{\SymbState\} \times
\Winning[\SymbState]$.

\thminv*

To prove this theorem, we will need a few auxiliary invariants:

\begin{theorem}\label{thm:aux}
  The following invariants stay true throughout the execution of the algorithm :
  \begin{enumerate}
  \setcounter{enumi}{5}
    \item $\WaitingExplore \cup \Explored \subseteq Reach(\SymbState_0)$ 
    \item $\forall \SymbState, \; \Depends[\SymbState] = \{ \SymbState' \in \Explored
    \; | \; \SymbState' \Rightarrow^{\trans} \SymbState \}$
    \item $\WaitingUpdate \subseteq \Explored$
  \end{enumerate}
\end{theorem}

\begin{proof} of \cref{thm:inv,thm:aux} 
  \begin{enumerate}
    \item The algorithm begins with the exploration of $ \TempSucc{\SymbState_0}$ at the end of which it is added to $\Explored$. At no point in the algorithm do we remove symbolic states from $\Explored$.
    \item A symbolic state $\SymbState$ is added to $\Explored$ at the end of its
    exploration. During the exploration, all symbolic states $\SymbState'$ such that
    $\SymbState \Rightarrow^{\trans} \SymbState'$ are observed by going through all transitions $\trans_s$ from $\SymbState$ and computing $ \TempSucc{Succ(\trans_s, \SymbState)}$ which corresponds to the symbolic state $\SymbState_s$ such that $\SymbState \Rightarrow^{\trans_s} \SymbState_s$. They are added to $\WaitingExplore$ if they are not already part of $\WaitingExplore \cup \Explored$. From there, we never remove a symbolic state from $\Explored$ and when we remove a symbolic state from $\WaitingExplore$, we do so at the end of its exploration and thus it is added to $\Explored$. Thus, nothing is ever removed from $\WaitingExplore \cup \Explored$.
    \item A symbolic state is added to $\Explored$ and removed from $\WaitingExplore$ at 
    the end of its exploration. After that, it is never removed from $\Explored$ and
    thus can't be added back in $\WaitingExplore$. It is thus explored only once during
    which we add $\SymbState$ to $\mathit{Win}[\SymbState]$
    if $\SymbState.l\in\TargetSet$.
    It cannot be updated beforehand, as $\WaitingUpdate \subseteq \Explored$ and afterwards
    $\mathit{Win}[\SymbState]$ can only be modified during an update. During an update,
    $\mathit{Win}[\SymbState]$ gets $\mathit{Win}[\SymbState] \cap \NewWin$ and thus
    can only grow. Thus, $\mathit{Win}[\SymbState]$ always contains $\SymbState$ if
    $\SymbState.l \in  \TargetSet$ after its exploration.
    \item For all symbolic states $\SymbState$, 
    before exploration we have $\mathit{Win}[\SymbState] = \emptyset \subseteq \Winning$.
    After exploration, but before the first update, if $\SymbState.l \in \TargetSet$,
    all state in $Symbstate$ are winning by having a target location and  $\mathit{Win}
    [\SymbState] = \SymbState \subseteq \Winning$.

    The operator $\SafePred(S_1, S_2)$ computes temporal predecessor of $S_1$ avoiding $S_2$. The temporal transition relation in the enhanced game arena is $\leadsto^{\delta_s} \; = \; (=_{\StateSet} \times \to^{\delta_s})$, which remains in the same symbolic state.
    Therefore:
    \begin{itemize}
\item    {\footnotesize $\SafePred( \Win \cup \WinningMoves(\Win) , \: \Uncontrollable(\Win) ) \cap (\{\SymbState\} \times \SymbState)$} {\scriptsize
    ${}=\SafePred( (\Win \cup \WinningMoves(\Win)) \cap (\{\SymbState\} \times \SymbState)  , \: \mathit{Uncontr}(\Win)\cap (\{\SymbState\} \times \SymbState))$}.
    \item $\WinningMoves(\Win) \cap (\{\SymbState\} \times \SymbState) = \underset{\trans_c \in \TransSet_c}{\bigcup} Pred(t_c, \Win) \cap (\{\SymbState\} \times \SymbState) = \underset{\trans_c \in \TransSet_c(\SymbState)}{\bigcup} Pred(t_c, \Win)$, where $\TransSet_c(\SymbState)$ are controllable actions from $\SymbState$.
    \item $\Uncontrollable(\Win) \cap (\{\SymbState\} \times \SymbState) = \underset{\trans_u \in \TransSet_u}{\bigcup} Pred(t_u, S \backslash \Win) \cap (\{\SymbState\} \times \SymbState) = \underset{\trans_u \in \TransSet_u(\SymbState)}{\bigcup} Pred(t_u, S \backslash \Win)$; $\TransSet_u(\SymbState)$ are uncontrollable actions from $\SymbState$.
    \end{itemize}    
    Thus, $\WinningMoves$ corresponds to $\WinningMoves(\Win) \cap (\{\SymbState\} \times \SymbState)$ 
    and $\Uncontrollable$ corresponds to $\Uncontrollable(\Win) \cap (\{\SymbState\} \times \SymbState)$,
    during an update of $\SymbState$.
    By monotony, since $\Win \subseteq \Winning$, we have
    $\WinningMoves(\Win)$ ${}\subseteq \WinningMoves(\Winning)$,
    $\Uncontrollable(\Winning) \subseteq \Uncontrollable(\Win)$.
    Thus:
    {\footnotesize
    \begin{align*}
    & \NewWin 
    \\ =\  &
    \SafePred( \Win \cup \WinningMoves(\Win) , \: \Uncontrollable(\Win) ) \cap (\{\SymbState\} \times \SymbState) \\ \subseteq\  &
    \SafePred( \Winning \cup \WinningMoves(\Winning), \: \Uncontrollable(\Winning) ) \cap (\{\SymbState\} \times \SymbState)
    \\ \subseteq\  &
    \SafePred( \Winning \cup \WinningMoves(\Winning) , \Uncontrollable(\Winning)) 
    \\ \subseteq\ & \Winning
    \end{align*}
    }
    We conclude that all new states added to $\mathit{Win}[\SymbState]$ are winning.
  \item
    When a Symbolic state is explored, it is added to $\WaitingUpdate$.
    
    During an update of $\SymbState$, we recall that 
    \begin{itemize}
      \item $\WinningMoves$ corresponds to $\WinningMoves(\Win) \cap (\{\SymbState\} \times \SymbState)$ and
      \item $Uncontrollable$ corresponds to $\Uncontrollable(\Win) \cap (\{\SymbState\} \times \SymbState)$ and
      \item $\SafePred(Win[\SymbState] \cup \WinningMoves, \Uncontrollable) \cap \SymbState$ corresponds to 
      {\footnotesize $\SafePred( \Win \cup \WinningMoves(\Win) , \: \Uncontrollable(\Win) ) \cap (\{\SymbState\} \times \SymbState)$}.
    \end{itemize}

    A symbolic state $\SymbState$ is removed from $\WaitingUpdate$ at the end of an update.
    Adding $\SafePred(Win[\SymbState] \cup \WinningMoves, \Uncontrollable) \cap \SymbState$
    to $\mathit{Win}[\SymbState]$, 
    corresponds to adding to $\Win$ the set\\
    {\footnotesize 
    $\SafePred( \Win \cup \WinningMoves(\Win) , \: \Uncontrollable(\Win) ) \cap (\{\SymbState\} \times \SymbState)$}. 
    
    We thus have 
    {\footnotesize
    $$ {{\SafePred( {\Win} \cup {\WinningMoves(\Win)} , \: \Uncontrollable(\Win))} \cap {\SymbState}}\subseteq {\Win}\enspace$$}

    Consequently, or simultaneously in the case of a self loop, the value of $\SafePred
    ( \Win \cup \WinningMoves(\Win) , \: \Uncontrollable(\Win) ) \cap \SymbState$
    may change only if the value of $\mathit{Win}[\SymbState']$ is modified where $\SymbState
    \Rightarrow^{\trans} \SymbState'$.
    In that case, $\Depends[\SymbState'] = \{ \SymbState'' \in \Explored \; | \; \SymbState'' \Rightarrow^{\trans} \SymbState' \}$ is added to $\WaitingUpdate$. Thus, $\SymbState$ is added back to $\WaitingUpdate$. 
\item The initial value of $ \WaitingExplore
\cup \Explored$ is $\{\TempSucc{\SymbState_0}\}$
which is included in $Reach(\SymbState_0)$. From there, a symbolic state $\SymbState'$
is added to $ \WaitingExplore \cup \Explored$ when it is reachable in one transition
from a symbolic state $\SymbState$ being explored. $\SymbState$ is thus currently in $\WaitingExplore \in \WaitingExplore \cup \Explored \subseteq Reach(\SymbState_0)$. Consequently, $\SymbState'$ is also a subset of $Reach(\SymbState_0)$.
\item $\forall \SymbState, \SymbState'$ symbolic states such that $\SymbState \Rightarrow_{\trans} \SymbState'$, if $\SymbState$ enters $\Explored$, it is at the end of its exploration during which it is added to $\Depends[\SymbState']$.
      From there, $\SymbState$ is never removed from either $\Explored$ or $\Depends[\SymbState']$.
\item A symbolic state enters $\WaitingUpdate$ either at the end of its exploration
when it is added to $\Explored$ or from a $\Depends[\SymbState']$. By the previous
point $\Depends[\SymbState'] \subseteq \Explored$.
  \end{enumerate}
\end{proof}

\algcorrect*

  \begin{proof}
    If the algorithm terminates, the waiting lists are empty. We then have : 
    \begin{itemize}
      \item $ \TempSucc{\SymbState_0}\in\Explored$.
      \item $\forall \SymbState \in \Explored, \forall \SymbState'$ such that $\SymbState
      \Rightarrow^{\trans} \SymbState': \SymbState' \in \Explored$.
    \end{itemize}
  
    We can then conclude that the union of the explored symbolic states covers $Reach(\SymbState_0)$.
    We also have, for all $\SymbState \in \Explored$:
    \begin{itemize}
      \item If $\SymbState.l \in \TargetSet$, then $\{\SymbState\} \times \SymbState \subseteq \Win $.
      \item $\Win \subseteq \Winning$.
      \item 
      \footnotesize{$\SafePred( \Win \cup \WinningMoves(\Win) , \: \mathit{Uncontrol}(\Win) ) \cap (\{\SymbState\} \times \SymbState) \subseteq \Win$.}
    \end{itemize}
  
    From there, we have the following:
    {\footnotesize
    \begin{align*}
     & \SafePred( \Win \cup \WinningMoves(\Win) , \: \Uncontrollable(\Win) )\\
   =\  &
    \underset{\SymbState \in \Explored}{\bigcup} \SafePred( \Win \cup
    \WinningMoves(\Win) , \: \mathit{Uncontrol}(\Win) ) \cap (\{\SymbState\} \times \SymbState) \\
    \subseteq\  & \Win 
    \end{align*}
    }
  Using the fixed point definition of winning states, we can conclude that $\Win$ is the
  set of winning states of $\StateSpace_{alg}$ over the symbolic states in $\Explored$ and thus $\underset{\SymbState \in \Explored}{\bigcup} Win[\SymbState]$ is the set of winning states in the initial game restricted to $\underset{\SymbState \in \Explored}{\bigcup} \SymbState$, which covers $Reach(\SymbState_0)$.
  
  Since all runs and histories of the game are defined from an initial state, they
  are contained in $Reach(\SymbState_0)$. Strategies of the game are defined on histories contained in $Reach(\SymbState_0)$, therefore they are strategies of the game restricted to $Reach(\SymbState_0)$ and an initial state is winning if and only if it is winning on $Reach(\SymbState_0)$.
  
  Then, the winning constraint on the parameters is the projection of the winning
  initial states on the parameters: 
  $$\WinningParam = \underset{\SymbState \in \Explored}{\bigcup} \ProjectParam{(Win[\SymbState] \cap \SymbState_0)}\enspace$$
  \end{proof}

\section{\imitator Model of the Production Cell}
An \imitator model of the 1-plate production cell can be seen in  \cref{fig:model_plate,fig:model_robot,fig:model_broadcaster}. The model is composed of three automata: the plate, the robot and the broadcaster. Arrows with same color signify a synchronized action in \imitator.
\label{app:prodcell}
\begin{figure}[h!]
  \centering
  \includegraphics[width=.61\textwidth]{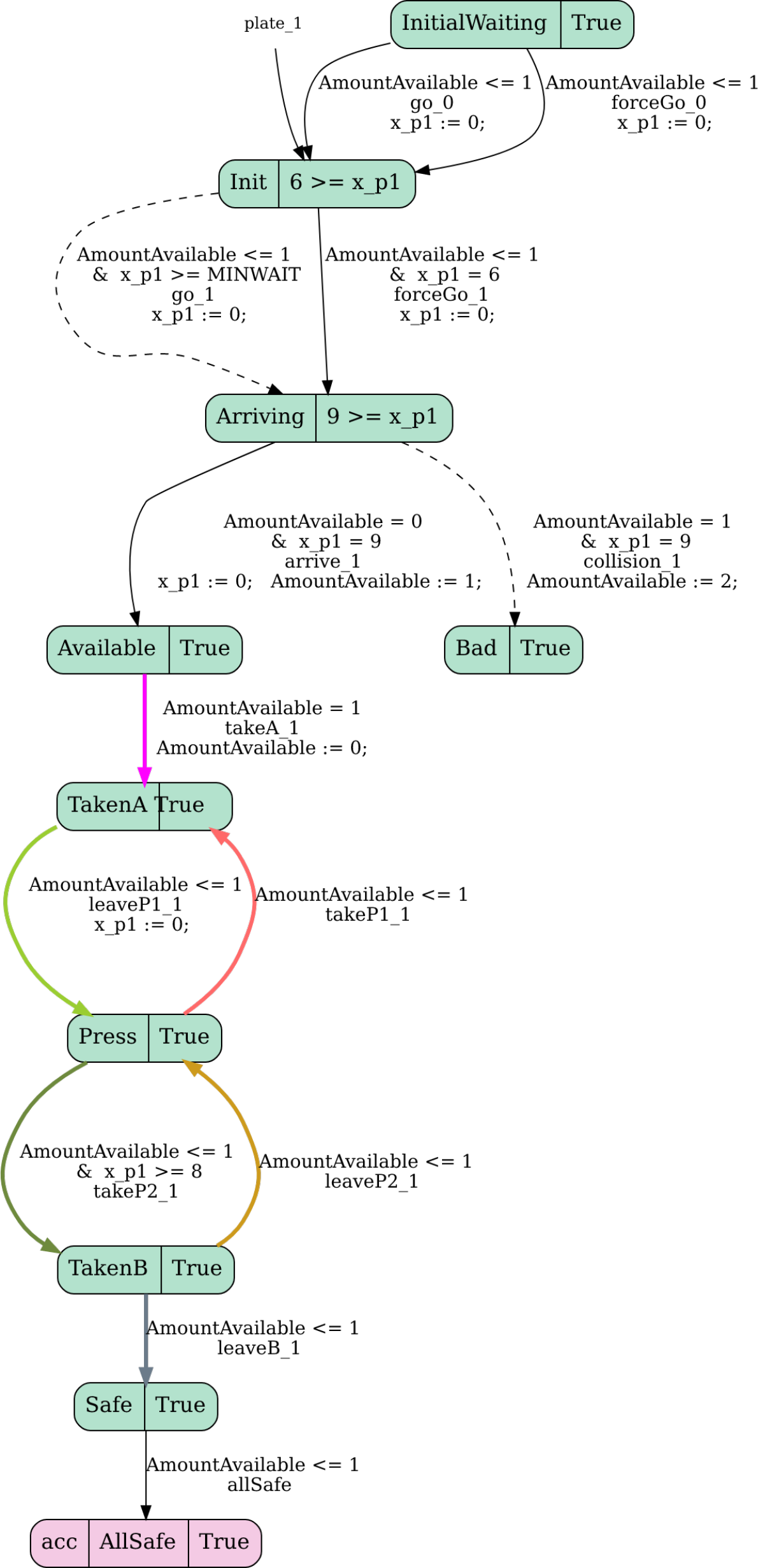}
  \caption{\imitator model for the Production Cell -- plate component}
  \label{fig:model_plate}
\end{figure}

\begin{figure}[t]
  \hspace{-2.5cm}
  \includegraphics[width=1.5\textwidth]{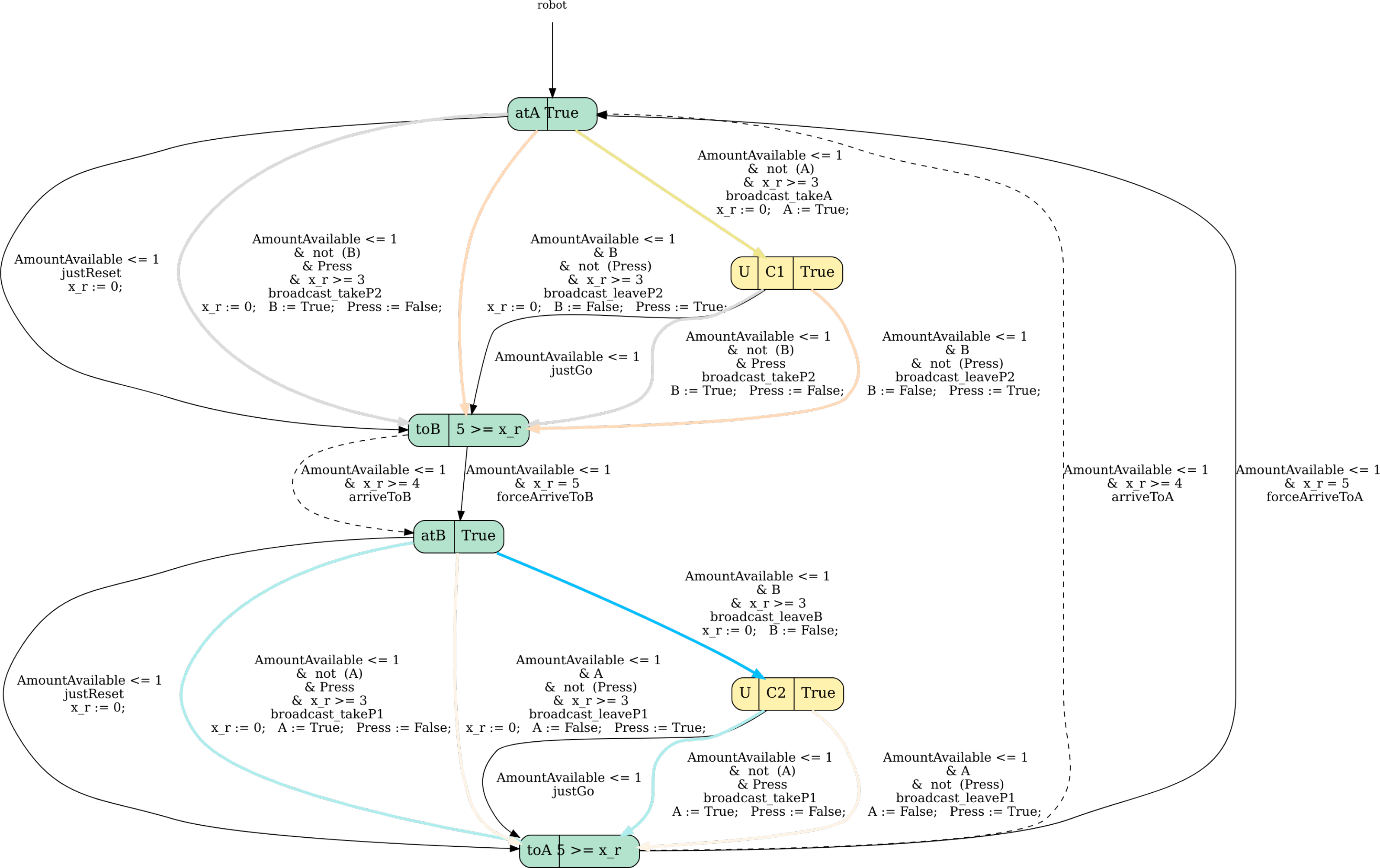}

  \bigskip
  \caption{\imitator model for the Production Cell -- robot component}
  \label{fig:model_robot}
\end{figure}

\bigskip

\begin{figure}[b]
  \hspace{-2.5cm}
  \includegraphics[width=1.5\textwidth]{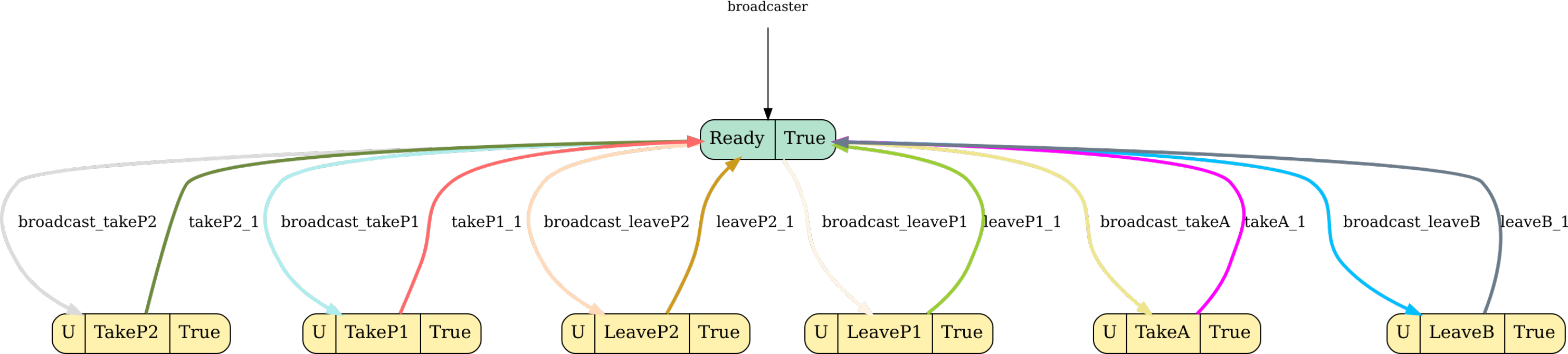}

  \bigskip
  \caption{\imitator model for the Production Cell -- broadcaster component\\ (auxiliary automaton to facilitate broadcast communication)}
  \label{fig:model_broadcaster}
\end{figure}

\end{document}